\documentclass[preprint]{revtex4}
\usepackage{enumerate}
\usepackage{graphicx}
\usepackage{dcolumn}
\usepackage{bm}
\usepackage{amsmath}
\usepackage{amsxtra}
\usepackage{amstext}
\usepackage{amssymb}
\usepackage{latexsym}

\usepackage{amscd,verbatim}
\usepackage[all]{xy}

\providecommand{\U}[1]{\protect\rule{.1in}{.1in}}
\newtheorem{theorem}{Theorem}

\newtheorem{proposition}[theorem]{Proposition}
\newtheorem{remark}[theorem]{Remark}

\newenvironment{proof}[1][Proof]{\noindent\textbf{#1.}}{\ \rule{0.5em}{0.5em}}
\begin{document}

%%%%%%%%%%%%%%%%%%%%%%%%TITLE%%%%%%%%%%%%%%%%%%%%%%%%
\title{
%\vskip-2.5truecm
%\rightline{\small{\tt ULB-TH/08-15}}
%\vskip2.5truecm
{Symmetry Group and Group Representations Associated to the Thermodynamic Covariance Principle (TCP)}
}
%%%%%%%%%%%%%%%%%%%%%%END_TITLE%%%%%%%%%%%%%%%%%%%%%%%

%%%%%%%%%%%%%%%%%%%%%%ADDRESSES%%%%%%%%%%%%%%%%%%%%%%
\author{Giorgio Sonnino$^{1,}{}^{2\star}$, Jarah Evslin$^3$, \\ Alberto Sonnino$^{4},$ Gy{$\ddot{\rm o}$}rgy Steinbrecher${}^{5}$, Enrique Tirapegui${}^{6}$}
\affiliation{
${}^1$ Department of Theoretical Physics and Mathematics, Universit{\'e} Libre de Bruxelles (U.L.B.), 
Campus Plaine C.P. 231, Bvd du Triomphe, 1050  Brussels - Belgium.\\ 
${}^2$ Royal Military School (RMS), Av. de la Renaissance 30, 1000 Brussels - Belgium.\\
${}^3$High Energy Nuclear Physics Group, Institute of Modern Physics,
Chinese Academy of Sciences, Lanzhou - China.\\
${}^4$ Department of Computer Science, University College London (UCL), \\ Gower St, WC1E 6BT London - United Kingdom.\\
${}^5$Physics Department - University of Craiova Str. A. I. Cuza 13 200585 Craiova - Romania.\\
${}^6$Departamento de F{\'i}sica, Facultad de Ciencias F{\'i}sicas y Mathem{\'a}ticas, \\Universidad de Chile, Casilla 487-3 Santiago de Chile - Chile.}
%%%%%%%%%%%%%%%%%%%%END_ADDRESSES%%%%%%%%%%%%%%%%%%%%%

%%%%%%%%%%%%%%%%%%%%%%ABSTRACT%%%%%%%%%%%%%%%%%%%%%%%
\begin{abstract}
We describe the Lie group and the group representations associated to the nonlinear Thermodynamic Coordinate Transformations (TCT). The TCT guarantee the validity of the Thermodynamic Covariance Principle (TCP) : {\it The nonlinear closure equations, i.e. the flux-force relations, everywhere and in particular outside the Onsager region, must be covariant under TCT}. In other terms, the fundamental laws of thermodynamics should be manifestly covariant under transformations between the admissible thermodynamic forces, i.e. under TCT. The TCP ensures the validity of the fundamental theorems for systems far from equilibrium. The symmetry properties of a physical system are intimately related to the conservation laws characterizing that system. Noether's theorem gives a precise description of this relation. We derive the conserved (thermodynamic) currents and, as an example of calculation, a simple system out of equilibrium where the validity of TCP is imposed at the level of the kinetic equations is also analyzed.


\noindent ${}^{\star}$ Email: gsonnino@ulb.ac.be
\end{abstract}

\maketitle
%%%%%%%%%%%%%%%%%%%%%END_ABSTRACT%%%%%%%%%%%%%%%%%%%%%%%

\section{Introduction - The Thermodynamical Field Theory (TFT) at a glance and scope of the present work}

%%%%%%%%%%%%%%%%%%%%%%%%%%%%%%%%%%%%%%%%%%%%%%%%%%%%%

In a previous work, one of us introduced a macroscopic theory for closure relations for systems out of Onsager${}^{\prime}$s region \cite{sonnino}-\cite{sonnino2}. The most important closure relations are the so-called transport equations, relating the dissipative fluxes to the thermodynamic forces that produce them. The latter is related to the spatial inhomogeneity and is expressed as gradients of the thermodynamic quantities. The study of these relations is the object of non-equilibrium thermodynamics. Indicating with $X^\mu$ and $J_\mu$ the thermodynamic forces and fluxes, respectively, the flux-force relations read
\begin{equation}\label{IR1}
J_{\nu}=\tau_{\mu\nu}(X)X^\nu
\end{equation}
\noindent where $\tau_{\mu\nu}(X)$ are the transport coefficients, and it is clearly put in evidence that the transport coefficients may depend on the thermodynamic forces. We suppose that all quantities involved in Eq.~(\ref{IR1}) are written in dimensionless form. In this equation, as in the remainder of this paper, the Einstein summation convention on the repeated indexes is adopted. Matrix $\tau_{\mu\nu}(X)$ can be decomposed into a sum of two matrices, one symmetric and the other skew-symmetric, which we denote with $g_{\mu\nu} (X)$ and $f_{\mu\nu}(X)$, respectively. The second law of thermodynamics requires that $g_{\mu\nu}$ be a positive-definite matrix. Note that, in general, the {\it entropy production}, which we denote by $\sigma$ with $\sigma=\tau_{\mu\nu}(X)X^\mu X^\nu=g_{\mu\nu}(X)X^\mu X^\nu$, may not be a bilinear expression of the thermodynamic forces (since the transport coefficients may depend on the thermodynamic forces). For conciseness, in the sequel we drop the symbol $X$ in $g_{\mu\nu}$ as well as in the skew-symmetric piece of the transport coefficients, $f_{\mu\nu}$, being implicitly understood that these matrices may depend on the thermodynamic forces. 

\noindent The aim of the theory in Ref.~\cite{sonnino} is to determine the nonlinear flux-force relations such that

\noindent $\bf A)$ respect the thermodynamic theorems for systems far from equilibrium; 

\noindent $\bf B)$ are covariant under the Thermodynamic Coordinate Transformations (TCT) i.e., the closure relations should be covariant under the transformations of the thermodynamic forces leaving unaltered both the entropy production, $\sigma$, and the Glansdorff-Prigogine dissipative quantity, $P$ [for the definition of $P$, see the forthcoming Eq.~(\ref{I1})]. As we shall see below, this covariance property is intimately related the concept of systems equivalent from the thermodynamic point of view. 

\noindent In addition to A) and B), the theory rests upon the following assumption:

\noindent $\bf C)$ Close to the steady-states, there exists a thermodynamic action, scalar under thermodynamic coordinate transformations, which is stationary for general variations in the transport coefficients and the affine connection. 

\noindent This theory, based on  A), B) and C), is referred to as {\it Thermodynamical Field Theory} (TFT). Let us now discuss the physical meaning of A), B) and C). 

$\bullet$ {\bf Constraint A)}. 

\noindent Among the thermodynamic theorems, which should be satisfied by systems out of equilibrium, we cite the well-known {\it Universal Criterion of Evolution} (UCE) \cite{prigogine1}-\cite{prigogine2}. In particular, when the system is close to the thermodynamic equilibrium the theorems valid in the Onsager regime (such as, for example, the {\it Minimum Entropy Production Theorem} (MEPT) \cite{prigogineR1}-\cite{prigogineR2}) should be recovered by the non-linear theory. Constraint A) allows introducing the so called {\it Space of the Thermodynamic Forces} (or, simply, the {\it Thermodynamic Space}) in the following manner (see Fig.~(\ref{Thermodynamic_Space}) and Ref.~\cite{sonnino}). The coordinates of the Thermodynamic Space are the thermodynamic forces, the metric is identified with the (positive-definite) piece of the transport coefficients and the parallel transport of a vector is made by the affine connection constructed in such a way that the Universal Criterion of Evolution (UCE) is automatically satisfied. We recall that the UCE is valid for systems out of equilibrium and subject to time-independent boundary conditions \cite{sonnino}. 
%%%%%%%%%%%%%%%%%%%%%%%%
\begin{figure*}[htb] 
\hspace{0cm}\includegraphics[width=7.5cm,height=7.5cm]{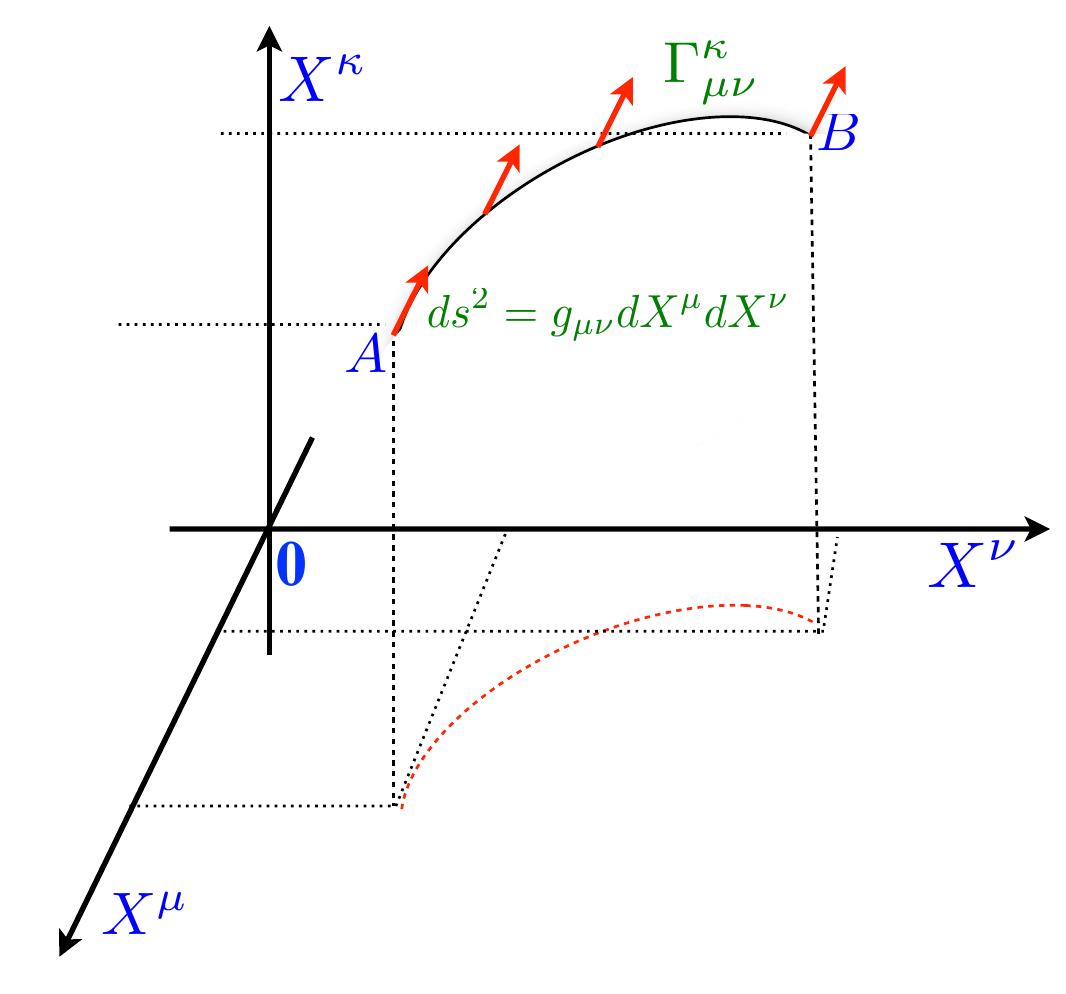}
\caption{ \label{Thermodynamic_Space} The thermodynamic space. The space is spanned by the thermodynamic forces. The metric tensor is identified with the symmetric piece, $g_{\mu\nu}$, of the transport coefficients, and the expression of the affine connection, $\Gamma^\kappa_{\mu\nu}$, is determined by imposing the validity of the Universal Criterion of Evolution. Note that the square of the length element, $ds^2={\mathbf ds}\cdot {\mathbf ds}$, is always a non-negative quantity for the second law of thermodynamics.}
\end{figure*}
%%%%%%%%%%%%%%%%%%%%%%%%

$\bullet$ {\bf Constraint B)}.

\noindent The main objective of this work is to analyze the symmetry underneath the Thermodynamic Covariance Principle. To this end, let us come back to the concept of equivalent systems from the thermodynamic point of view. This concept was originally introduced by Th. De Donder and I. Prigogine, and it has been recently deeply investigated and revised in Refs~\cite{sonnino}, \cite{sonnino1}. This statement stems from the Einstein formula linking the {\it probability of a fluctuation}, $\mathcal W$, with the {\it entropy production strength}, $\Delta_I S$, associated with the fluctuations from the non-equilibrium steady state. Denoting by $\xi_i$ ($i=1\cdots m$) the $m$ deviations of the thermodynamic quantities from their equilibrium value, Prigogine proposed that the probability distribution of finding a state in which the values $\xi_i$ lie between $\xi_i$ and $\xi_i+d\xi_i$ is given by \cite{prigogineR2}
\begin{equation}\label{i3a}
\mathcal{W}=W_0\exp[\Delta_{\rm I}  S/k_B]\qquad\quad
{\rm where}\qquad \Delta_{\rm I}   S=\int_E^F d_{\rm I} s\quad  {\rm ;}\quad \frac{d_{\rm I}  s}{dt}\equiv\int_\Omega\sigma dv
\end{equation}
\noindent Here, $k_B$ is the Bolzmann constant and $W_0$ is a normalization constant that ensures the sum of all probabilities equals one. In addition, $dv$ is a (spatial) volume element of the system, and the integration is over the entire space $\Omega$ occupied by the system in question. $E$ and $F$ indicate the equilibrium state and the state to which a fluctuation has driven the system, respectively. We note that the probability distribution (\ref{i3a}) remains unaltered for flux-force transformations leaving invariant the entropy production. On the basis of the above observations, and other concrete examples analyzed in \cite{prigogineR1}-\cite{prigogineR2}, Th. De Donder and I. Prigogine formulated, for the first time, the concept of {\it equivalent systems from the thermodynamical point of view}. For Th. De Donder and I. Prigogine, {\it thermodynamic systems are thermodynamically equivalent if, under transformation of fluxes and forces, the bilinear form of the entropy production remains unaltered, i.e.,} $\sigma=\sigma'$. \cite{prigogineR2}. Hence, in classical textbooks on linear and nonlinear
irreversible thermodynamics the concepts of {\it equivalence of thermodynamic systems} is formulated only in terms of invariance of the entropy production under the thermodynamic force transformations (see, for example, \cite{degroot}, \cite{ottinger}). 

\noindent However, the condition of the invariance of the entropy production is not sufficient to ensure the equivalence character of the  two descriptions $(J_\mu , X^\mu)$ and $(J'_\mu , X'^\mu )$. Indeed, we can convince ourselves that there exists a large class of transformations such that, even though they leave unaltered the expression of the entropy production, they may lead to certain paradoxes to which Verschaffelt has called attention \cite{verschaffelt}-\cite{davies}. In addition, the above De Donder-Prigogine definition is unable to determine, univocally, the most general class of the thermodynamic force transformations able to ensure the equivalence among thermodynamic systems (see, \cite{degroot}). These obstacles may be overcome if one takes into account one of the most fundamental and general theorems valid in thermodynamics of irreversible processes : the {\it Universal Criterion of Evolution} (UCE) \cite{prigogine1}, \cite{prigogine2}. Without using neither the Onsager reciprocal relations nor the assumption that the phenomenological coefficients (or linear phenomenological laws) are constant, for time-independent boundary conditions, when the system relaxes towards a {\it stable steady state}, the dissipative quantity $P$, defined as \cite{prigogine1}-\cite{prigogine2}
\begin{equation}\label{I1}
P\equiv\int_\Omega J_\mu\frac{d X}{d t}^\mu\ dV\leq 0
\end{equation}
\noindent is always a negative quantity, being the negative sign due to the {\it stability of the system}. In Eq.~(\ref{I1}), $\Omega$ is the volume occupied by the system and $dV$ the volume-element, respectively. In addition
\begin{equation}\label{IR1a}
 \int_\Omega J_\mu\frac{dX}{dt}^\mu dV= 0 \quad {\rm at\ the\ steady\ state} 
 \end{equation}
\noindent hence, quantity $P\equiv\int_\Omega J_\mu\frac{dX^\mu}{dt} dV$ is a sort of intrinsic quantity of a dissipative system, and it may be referred to as {\it the Glansdorff-Prigogine dissipative quantity}. In Refs~\cite{sonnino}, \cite{sonnino1}, it is shown that to ensure the equivalent character of the two descriptions, $\{X^\mu\}$ and $\{X^{\prime\mu}\}$, it is not sufficient to require that the entropy production of the system, $\sigma=g_{\mu\nu}X^\mu X^\nu$, is invariant under the flux-force transformation, but we should also require that the Glansdorff-Prigogine dissipative quantity remains invariant under transformation of the thermodynamic forces $\{X^\mu\}\rightarrow\{X^{\prime\mu}\}$. Indeed, in general, we are free to make our choice of the set of the thermodynamic forces. For example, if we analyze the case of heat conduction in non-expanding solid, where chemical reactions take place simultaneously, we can choose as thermodynamic forces the (dimensionless) chemical affinities over Temperature and the (dimensionless) gradient of the inverse of the Temperature. It is quite possible to modify the definition of the thermodynamic forces by taking a linear combination of the chemical affinities over temperature and the gradient of the inverse of Temperature \cite{jqc}. Of course, these two representations of the thermodynamic forces are equivalent only if the transformation between these two set of thermodynamic forces leaves unaltered the expression of the entropy production {\it and} it preserves the negative sign of the quantity $P$ (otherwise the UCE will be violated due to the mathematical transformation). 

\noindent Magnetically confined Tokamak-plasmas are a typical example of thermodynamic systems, out of Onsager's region, where the equivalence between two different choices of the thermodynamic forces is warranted only if both the entropy production and the Glansdorff-Prigogine dissipative quantity $P$, defined above, remain unaltered under transformation of these thermodynamic forces \cite{sonnino3}.  

\noindent By summarizing, the (admissible) thermodynamic forces should satisfy the following two conditions:

\noindent {\bf 1)} {\it The entropy production, $\sigma$, should be invariant under transformation of the thermodynamic forces $\{X^\mu\}\rightarrow\{X^{\prime\mu}\}$} and

\noindent {\bf 2)} {\it The Glansdorff-Prigogine dissipative quantity, $P$, should also remain invariant under the forces transformations $\{X^\mu\}\rightarrow\{X^{\prime\mu}\}$}.

\noindent Condition {\bf 2)} stems from the fact that

\noindent {\bf 2a)} {\it The steady state should be transformed into a steady state;} 

\noindent and

\noindent {\bf 2b)} {\it The stable steady state should be transformed into a stable state state, with the same {\it degree} of stability.}

\noindent This kind of transformations may be referred to as the {\it Thermodynamic Coordinate Transformations} (TCT). The forthcoming Eqs~(\ref{topology1}), in Sec.~\ref{topology}, provide with the most general class of TCT satisfying conditions {\bf 1)} and {\bf 2)}. As we shall prove in Section~\ref{topology}, the TCT are not, simply, transformations written in a projective form, but they form a nontrivial bundle whose base is the projective space and whose fiber is the space of maps from the projective space to the non-vanishing reals.

\noindent The thermodynamic equivalence principle leads, naturally, to the following {\it Thermodynamic Covariance Principle} (TCP) : {\it The nonlinear closure equations, {\it i.e.} the flux-force relations, must be covariant under TCT} \cite{sonnino1}. The essence of the TCP is the following. The equivalent character between two representations is warranted if, and only if, the fundamental thermodynamic equations are covariant under TCT. Loosely speaking, the covariant formalism warrants that the fundamental laws of thermodynamics, for instance the flux-force closure equations, remain invariant (or better, covariant) under transformations of the (admissible) thermodynamic forces. This is the correct mathematical formalism to ensure the equivalence between two different representations. Note that the TCP is trivially satisfied by the closure equations valid in the Onsager region. We mention that the linear version of the TCT is actually widely used for studying transport processes in Tokamak-plasmas (see, for examples, the papers cited in the book \cite{balescu2}).

$\bullet$ {\bf Assumption C)}.

\noindent According to assumption C), there exists an action which is stationary with respect to arbitrary variation of the transport coefficient and the affine connection. To avoid misunderstanding, while it is correct to mention that this postulate affirms the possibility of deriving the nonlinear closure equations by a variational principle, it does not state that the expressions and theorems obtained from the solutions of these equations can also be derived by a variational principle. In particular the Universal Criterion of Evolution {\it cannot} be derived by a variational principle. 

\noindent In the framework of the Thermodynamical Field Theory (TFT) introduced by one of us \cite{sonnino}, \cite{sonnino2} one can find the expression of the thermodynamic action \cite{sonnino}:
\begin{equation}\label{pa5}
I=\int\Bigl[ R-(\Gamma^\lambda_{\alpha\beta}-
{\tilde\Gamma}^\lambda_{\alpha\beta})S^{\alpha\beta}_{\lambda}
\Bigr]\sqrt{g}\ \! {d^{}}^n\!X
\end{equation}
\noindent with ${d^{}}^n\!X$  denoting an infinitesimal volume element in the space of the thermodynamic forces and $g$ the determinant of $g_{\mu\nu}$ (see \cite{sonnino}). In addition, $R$ is the curvature of the space. 

\noindent Action (\ref{pa5}) is derived by imposing that \cite{sonnino}
\vskip0.3truecm
\noindent {\bf i)} it is constructed with the two pieces of the transport coefficients, $g_{\mu\nu}$ and $f_{\mu\nu}$, the affine connection $\Gamma^\mu_{\alpha\beta}$ and {\it only} with the first-order derivatives of these fields;

\noindent {\bf ii)} it is invariant under TCT. This constraint ensures the validity of the TCP (i.e., the closure equations should be covarianat under TCT);

\noindent {\bf iii)} it is stationary when $\Gamma^\mu_{\alpha\beta}=\tilde{\Gamma}^\mu_{\alpha\beta}$ i.e., the action is stationary only when the affine connection coincides with the expression able to satisfy (automatically) the UCE; 

\noindent {\bf iv)} metric $g_{\mu\nu}$ and the skew-symmetric piece of the transport coefficients $f_{\mu\nu}$ tend to the Onsager matrices as the thermodynamic system approaches equilibrium.
\vskip0.3truecm
\noindent As shown in Ref.~\cite{sonnino}, in general, action (\ref{pa5}) is a quite complex. However, in case of magnetically confined plasmas (which is the case analyzed in this work) the skew symmetric pieces of the transport coefficients $f_{\mu\nu}$ are zero, and the action simplifies notably because the terms appearing in Eq.~(\ref{pa5}) reduce to \cite{sonnino}
\begin{eqnarray}\label{pa7}
&&R=R_{\mu\nu}g^{\mu\nu}\\
&&R_{\mu\nu}=\Gamma^\kappa_{\nu\kappa ,\mu}-\Gamma^\kappa_{\nu\mu ,\kappa}\!+\!\Gamma^\kappa_{\nu\lambda}\Gamma^\lambda_{\kappa\mu}\!-\!\Gamma^\kappa_{\nu\mu}\Gamma^\lambda_{\kappa\lambda}\nonumber\\
&&S_\lambda^{\mu\nu}=\Psi_{\lambda\alpha}^\nu g^{\nu\alpha}\!+\!\Psi_{\lambda\alpha}^\mu g^{\mu\alpha}\!-\!\frac{1}{2}\Psi_{\alpha\beta}^\mu g^{\alpha\beta}\delta^\nu_\lambda\!-\!\frac{1}{2}\Psi^\nu_{\alpha\beta}g^{\alpha\beta}\delta^\mu_\lambda\nonumber\\
&&\Psi^\mu_{\alpha\beta}=\frac{1}{2\sigma}X^\kappa X^\mu g_{\alpha\beta,\kappa}\!-\!\frac{X^\kappa X^\lambda}{2(n+1)\sigma}\left(\delta^\mu_\alpha g_{\beta\kappa , \lambda}\!+\!\delta^\mu_\beta g_{\alpha\kappa , \lambda}\right)\nonumber\\
&&{\tilde\Gamma}^\mu_{\alpha\beta}=
\begin{Bmatrix} 
\mu \\ \alpha\beta \end{Bmatrix}
\!+\!\frac{1}{2\sigma}X^\mu X^\kappa g_{\alpha\beta ,\kappa}\!-\!\frac{X^\kappa X^\lambda}{2(n+1)\sigma}\left(\delta^\mu_\alpha g_{\beta\kappa ,\lambda}\!+\!\delta^\mu_\beta g_{\alpha\kappa ,\lambda}\right)\nonumber\\
&&\begin{Bmatrix} 
\mu \\ \alpha\beta \end{Bmatrix}
=\frac{1}{2}g^{\mu\lambda}\left(g_{\lambda\alpha ,\beta}\!+\!g_{\lambda\beta,\alpha}\!-\!g_{\alpha\beta,\lambda}\right)\nonumber\\
&&\Delta\Gamma^\mu_{\alpha\beta}\equiv\Gamma^\mu_{\alpha\beta}-\begin{Bmatrix} 
\mu \\ \alpha\beta \end{Bmatrix}=\!\frac{1}{2\sigma}X^\mu X^\kappa g_{\alpha\beta ,\kappa}\!-\!\frac{X^\kappa X^\lambda}{2(n+1)\sigma}\left(\delta^\mu_\alpha g_{\beta\kappa ,\lambda}\!+\!\delta^\mu_\beta g_{\alpha\kappa ,\lambda}\right)
\end{eqnarray}
\noindent with $g^{\mu\nu}$ denoting the inverse matrix of $g_{\mu\nu}$ and $","$ stands for the partial derivative with respect to the thermodynamic forces. The most general expressions for $S_\lambda^{\mu\nu}$, $\Psi^\mu_{\alpha\beta}$ and ${\tilde\Gamma}^\mu_{\alpha\beta}$ (and, then, for $\Delta\Gamma^\mu_{\alpha\beta}$), valid when $f_{\mu\nu}\neq 0$, can be found in Ref.~\cite{sonnino}. It can be shown that action (\ref{pa5}) is stationary when the {\it thermodynamic affine connection} $\Gamma^\mu_{\alpha\beta}$ is equal to ${\tilde\Gamma}^\mu_{\alpha\beta}$, with ${\tilde\Gamma}^\mu_{\alpha\beta}$ having the expression given by the r.h.s. of the fifth equation of Eqs~(\ref{pa7}) (only when $f_{\mu\nu}=0$). The sixth equation of Eqs~(\ref{pa7}) corresponds to the Levi-Civita affine connection valid in General Relativity. As shown by the seventh equation of Eqs~(\ref{pa7}), the thermodynamic affine connection differs widely from the Levi-Civita connection \footnote{Note that $\Delta\Gamma^\mu_{\alpha\beta}$ is a mixed third-order tensor under TCT.}. These two affine connections tend to identify each other only for very large values of the entropy production ($\sigma\gg 1$) \footnote{As shown in Ref.~\cite{sonnino}, the geometries of the General Relativity (GR) and the Thermodynamical field Theory (TFT) are widely different. Indeed, in the GR the geometry is pseudo-Riemannian, the field is symmetric and the affine connection is given by the Levi-Civita expression. In addition, the GR rests upon the validity of the {\it General Covariance Principle} in the space-time and on the validity of the {\it Equivalence Principle}. In the GR, the Universal Criterion of Evolution is not satisfied. In the TFT, the geometry is non-Riemannian, the field is asymmetric and the thermodynamic affine connection is given by ${\Gamma^\mu_{\alpha\beta}=\tilde\Gamma}^\mu_{\alpha\beta}$. The TFT rests upon the validity of the (special) covariance principle TCP and on the validity of the Universal Criterion of Evolution. In the TFT, the Equivalence Principle is not satisfied. For more details, see the annex of \cite{sonnino}.}. Note that the first piece of Lagrangian, i.e., $R\sqrt{g}$, is due to geometry and it is a generic contribution which appears whenever the curvature of the space is constructed through the Riemann tensor (e.g., space-time, thermodynamic space, etc). Indeed, the vanishing divergence of the tensor derived by the term $R\sqrt{g}$ reflects the geometric unchangeable property which comes from the theorem that the {\it boundary of a boundary is zero}. As known, this theorem gives the geometrical interpretation of the algebraic Bianchi identity \footnote{If we consider an infinitesimal cubical coordinate volume in the space of the thermodynamic forces, when a generic vector ${\mathbf A}^\mu$ is parallel transported around all the six surfaces of the cube, all the edges are traversed twice, once in each direction. These displacements have signs depending upon the direction in which an edge is traversed, so all the displacements add up to zero \cite{wheeler}.}. On the other hand, the tensor derived by the term $(\Gamma^\lambda_{\alpha\beta}-
{\tilde\Gamma}^\lambda_{\alpha\beta})S^{\alpha\beta}_{\lambda}\!\sqrt{g}$ reflects physics. The physical meaning of the tensor, constructed by this piece of Lagrangian, rests upon the Noether current. In our case (i.e., the TFT), this theorem reflects the required symmetry expressing the invariance of the Lagrangian under TCT. The tensor derived by the Noether current is the {\it source} of the thermodynamic space and, as shown in \cite{sonnino}, it vanishes in the Onsager region \footnote{Indeed, it is possible to show that the {\it source term} of the thermodynamic space is the second order thermodynamic-tensor $T_{\mu\nu}=-S_\lambda^{\alpha\beta}\frac{\delta{\tilde{\Gamma}}^\lambda_{\alpha\beta}}{\delta g^{\mu\nu}}$ \cite{sonnino}, with $\frac{\delta{\tilde{\Gamma}}^\lambda_{\alpha\beta}}{\delta g^{\mu\nu}}$ denoting the variation of $\tilde{\Gamma}^\lambda_{\alpha\beta}$ with respect to $g^{\mu\nu}$.}.

\noindent By imposing the stationary of action (\ref{pa5}) with respect to small variations of the transport coefficients, we get the non-linear transport equations \cite{sonnino}. These equations tend to the Onsager transport equations when the system approaches equilibrium. We mention that it is possible to prove that, in the {\it weak field approximation}, i.e. when $g_{\mu\nu}(X)\simeq L_{\mu\nu}+h_{\mu\nu}(X)$, with $L_{\mu\nu}$ and $h_{\mu\nu}(X)$ denoting the Onsager transport coefficients matrix and the (weak) perturbation of the Onsager matrix respectively, and {\it for very large values of the entropy production} ($\sigma\gg 1$), we have \cite{sonnino}
\begin{equation}\label{pa8}
\Gamma^\kappa_{\mu\nu}={\tilde \Gamma}^\kappa_{\mu\nu}=\frac{1}{2}L^{\kappa\eta}(h_{\mu\nu ,\nu}+h_{\nu\eta ,\mu}-h_{\mu\nu , \eta})+h.o.t.
\end{equation}
\noindent where $h.o.t.$ stands for {\it higher order terms}. Clearly, a transport theory without knowledge of microscopic dynamical laws cannot be developed. Transport theory is only but an aspect of non-equilibrium statistical mechanics, which provides the link between micro-level and macro-level. This link appears indirectly in the {\it unperturbed} matrices, i.e., the $L^{\mu\nu}$ (and the $f_0^{\mu\nu}$) coefficients used as an input in the equations. These coefficients, which depend on the specific material under consideration, have to be calculated in the usual way by kinetic theory. The perturbation fields (i.e., the corrections to the Onsager transport coefficients) $h_{\mu\nu}(X)$ depend on the thermodynamic forces and for $\sigma\gg 1$ they are solutions of the equations \cite{sonnino}
\begin{equation}\label{pa9}
L^{\lambda\kappa}\frac{\partial^2 h_{\mu\nu}}{\partial X^\lambda X^\kappa}+L^{\lambda\kappa}\frac{\partial^2 h_{\lambda\kappa}}{\partial X^\mu X^\nu}-L^{\lambda\kappa}\frac{\partial^2 h_{\lambda\nu}}{\partial X^\kappa X^\mu}-L^{\lambda\kappa}\frac{\partial^2 h_{\lambda\mu}}{\partial X^\kappa X^\nu}=0+h.o.t.
\end{equation}
\noindent For $\sigma\sim 1$, Eqs~(\ref{pa8}) and (\ref{pa9}) loose of validity and the correct equations become much more complex. Note that for $\sigma\ll 1$ we enter into the Onsager regime. Eqs~(\ref{pa9}) should be solved with the appropriate boundary conditions. Concrete examples can be found in Refs~\cite{sonnino2}, \cite{jqc}, where the nonlinear thermoelectric effect and chemical reactions out of Onsager region are analyzed in detail. In these cases, the boundary conditions are obtained by imposing that, for very large values of the gradient of the inverse of the Temperature and of the applied electric field, the electrical and heat fluxes and the chemical flows have no privileged directions in the thermodynamic space \cite{sonnino2}, \cite{jqc}. As mentioned above, the Onsager matrix $L^{\mu\nu}$ is derived by kinetic theory and introduced, as an input, into Eqs~(\ref{pa9}). It is worth mentioning that for the case of chemical reactions the solution of Eqs~(\ref{pa9}), subject to the appropriate boundary conditions, coincides, exactly, with {\it De Donder's law of mass} \cite{degroot}, \cite{jqc}. 

\noindent The present work is organized as follows. In Section~\ref{topology}, we provide the topological description of the group of TCT. In Section~\ref{algebra} we show that the TCT-group may be split as a semidirect product of two subgroups where the first one is a normal (abelian) subgroup of the TCT-group. The mathematical details related to this demonstration can be found in the Appendix. In Section~\ref{noether} we obtain the expression of Noether's current for the case of Tokamak-plasmas in fully collisional transport regime. In Section~\ref {kinetic} we derive the collisional operator, which guarantees that the Thermodynamic Covariance Principle (TCP) is satisfied by the closure transport relations (i.e., the flux-force relations). We conclude by showing a comparison between an experimental profile of radial heat loss for FTU (Frascati Tokamak Upgrade)-Plasmas against the theoretical predictions obtained by the linear (Onsager) and the non-linear (TCP) theories. Concluding remarks are reported in Section \ref{conclusions}.

\section{Topological Description of the TCT-Group}\label{topology}

\subsection{The construction}

As shown in Ref.~\cite{sonnino}, the TCT are given by
\begin{equation}\label{topology1}
X^\mu\rightarrow X^{\prime\mu}=X^1 F^\mu\left(\frac{X^2}{X^1},\frac{X^3}{X^2},...,\frac{X^n}{X^{n-1}}\right) .
\end{equation}
where $F^\mu$ are {\it arbitrary functions} of variables $X^j/X^{j-1}$ with ($j=2,\dots, n$). We demand that the $n$ functions $F^\mu$ be smooth, so that the TCT preserve equations satisfied by the derivatives of thermodynamic quantities, and also that the transformation be nondegenerate with a smooth inverse, so that the transformed theory contain all of the information of the original theory.  The nondegenerate property is also a necessary and sufficient condition for the finiteness of the transformed transport coefficients, even though it implies that the $F^\mu$ themselves may sometimes diverge. For example, from the transformation
\[
X^{\prime 1}=X^2,\ \ X^{\prime 2}=X^1
\]
one obtains
\[
F^1(X^2/X^1)=X^2/X^1,\ \ F^2(X^2/X^1)=1
\]
showing that $F^1(X^2/X^1)$ diverges at $X^1=0$, whereas the $X^{\prime\mu}$ are always finite. 

\noindent The space of thermodynamic forces, linear combinations of $\{X^1,...,X^n\}$, is the real Euclidean space $\mathbb{R}^n$.  On the other hand the ratios $\{X^\mu/X^{\mu-1}\}$ are coordinates for a different space, the real projective space $\mathbb{RP}^{n-1}$, which is defined to be the quotient of  $\mathbb{R}^n$ minus the origin by the scaling map $X^\mu\rightarrow \alpha X^\mu$ where $\alpha$ is any nonzero real number.  Note that some of the $X^\mu$ may vanish, removing the origin simply implies that not all of the $X^\mu$ vanish simultaneously. Fig.~(\ref{Proj_Space3}) illustrates a space, which is diffeomorphic to the $\mathbb{RP}^{n-1}$.
%%%%%%%%%%%%%%%%%%%%%%%%
\begin{figure*}[htb] 
\hspace{0cm}\includegraphics[width=7.7cm,height=8cm]{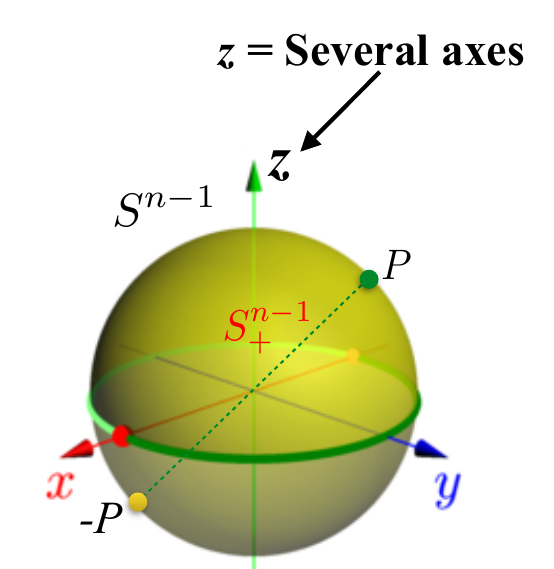}
\caption{ \label{Proj_Space3} The Projective Space $\mathbb{RP}^{n-1}$ is diffeomorphic to $S_+^{n-1}$ made by the Upper hemisphere + Half equator (without the red and yellow points) + the Red point.}
\end{figure*}
%%%%%%%%%%%%%%%%%%%%%%%%
\noindent Observe that $X^{\prime\mu}$ is an arbitrary smooth, degree 1 function of the $X$'s with the property that $X\rightarrow X^\prime$. is invertible.  This implies that   $X^{\prime\mu}/X^{\prime\mu-1}$ is an arbitrary degree 0 function with these same properties.  Now \{$X^{\prime\mu}/X^{\prime\mu-1}$\} are again coordinates of $\mathbb{RP}^{n-1}$.   The fact that $X^{\prime\mu}/X^{\prime\mu-1}$ is degree zero implies that it is invariant under the transformation $X^\mu\rightarrow \alpha X^\mu$ and so the map $X^{\prime\mu}/X^{\prime\mu-1}$ is in fact a map from $\mathbb{RP}^{n-1}\rightarrow \mathbb{RP}^{n-1}$
\[
\mathbb{RP}^{n-1}\rightarrow \mathbb{RP}^{n-1}:\frac{X^\mu}{X^{\mu-1}}\mapsto  \frac{X^{\prime\mu}}{X^{\prime\mu-1}}=\frac{F^\mu\left(\frac{X^2}{X^1},\frac{X^3}{X^2},...,\frac{X^n}{X^{n-1}}\right)}{F^{\mu-1}\left(\frac{X^2}{X^1},\frac{X^3}{X^2},...,\frac{X^n}{X^{n-1}}\right)} .
\]
So we have learned that the TCT yields a map from $\mathbb{RP}^{n-1}$ to itself.  Furthermore, the invertibility condition implies that this map is invertible and the smooth inverse condition implies that this map is a diffeomorphism.  Thus every TCT defines a diffeomorphism of $\mathbb{RP}^{n-1}$ to itself.

\noindent At this point it is tempting to conclude that the group of TCTs is just the group $diff(\mathbb{RP}^{n-1})$ of such diffeomorphisms.  However this is not quite true, because the ratios $X^{\prime\mu}/X^{\prime\mu-1}$ do not contain all the information in the $X^{\prime\mu}$. To reconstruct all the $X^{\prime\mu}$ from the ratios, one also needs to know, for example, $X^{\prime 1}$ or equivalently the real-valued function $F^1:\mathbb{RP}^{n-1}\rightarrow\mathbb{R}$, which intuitively gives the overall scale dependence of the TCT.  Therefore the group, $G$, of TCTs is a product of  $diff(\mathbb{RP}^{n-1})$ with the multiplicative group of maps from $\mathbb{RP}^{n-1}$ to the nonvanishing reals $\mathbb{R}^\times$ where the nonvanishing condition is needed to ensure nondegeneracy.   

\subsection{A Subtlety}

This is the right answer locally.  Globally there is one subtlety, we have double counted the map which flips the sign of all of the forces $X$.  It was the element $\alpha=-1$ which we have quotient when constructing $\mathbb{RP}^{n-1}$ from $\mathbb{R}^n$.  More precisely, $\mathbb{RP}^{n-1}$ can be constructed from $\mathbb{R}^n$ minus the origin in two steps.  First quotient by the maps $X\mapsto \alpha X$ with $\alpha$ positive, yielding the sphere $S^{n-1}$, and then quotient by $\alpha=-1$ yielding $\mathbb{RP}^{n-1}$.  This second action, whose quotient maps $S^{n-1}$ to $\mathbb{RP}^{n-1}$, has the same action on the X's as the map $-1$ in $\mathbb{RP}^{n-1}\rightarrow\mathbb{R}^\times$.   How does this double counting affect $G$?

\noindent Given a TCT one may calculate the ratio $\{X^{\prime\mu}/X^{\prime\mu-1}\}$. Since $G$ is the group of TCTs while $diff(\mathbb{RP}^{n-1})$ is the group of maps $\{X^{\prime\mu}/X^{\prime\mu-1}\}$, there must exist a projection $G\rightarrow diff(\mathbb{RP}^{n-1})$.  The argument above implies that the kernel of this projection is the space of nonvanishing maps from $\mathbb{RP}^{n-1}$ to $\mathbb{R}^\times$.  Therefore the group $G$ of TCTs is a bundle whose base is $diff(\mathbb{RP}^{n-1})$ and whose fiber is the space of maps $\mathbb{RP}^{n-1}\rightarrow\mathbb{R}^\times$.  Which bundle is it?

\noindent When traversing a noncontractible loop in $\mathbb{RP}^{n-1}$, which necessarily lifts in $S^{n-1}$ to a path between two antipodal points, the sign of the $\mathbb{R}^\times$ must change.  This means that the group of scalings $\mathbb{RP}^{n-1}\rightarrow\mathbb{R}^\times$ is nontrivial fibered over $diff(\mathbb{RP}^{n-1})$ such that, upon traversing the nontrivial cycle once, the sign of $\mathbb{R}^\times$ is inverted. 

\noindent Assembling all these arguments we arrive at our final result. The group G of TCTs  is the nontrivial bundle of the maps $\mathbb{RP}^{n-1}\rightarrow\mathbb{R}^\times$ over $diff(\mathbb{RP}^{n-1})$. There is a simple mathematical formulation for this group $G$.  Let $P:\mathbb{R}^n\backslash\{0\}\rightarrow \mathbb{RP}^{n-1}$ be the quotient $X\sim \alpha X$ which defines the real projective space $\mathbb{RP}^{n-1}$. Then the group $G$ of TCTs is the group of maps $f:\mathbb{RP}^{n-1}\rightarrow \mathbb{R}^n\backslash\{0\}$ such that $P\circ f:\mathbb{RP}^{n-1}\rightarrow\mathbb{RP}^{n-1}$ is a diffeomorphism.  Note that given  a TCT, $F:\mathbb{R}^n\rightarrow\mathbb{R}^n$  is given by $F(0)=0$ and away from the origin $F=f\circ P$. This construction is summarized in the commutative diagram
\begin{equation}\label{topology2}
\xymatrix{
\mathbb{R}^n\backslash\{0\}\ar[d]_{P}\ar[r]^{f\circ P}&\mathbb{R}^n\backslash\{0\}\ar[d]_{P}\\
\mathbb{RP}^{n-1}\ar[ur]^{f}\ar[r]^{P\circ f}&\mathbb{RP}^{n-1}
} 
\end{equation}
\noindent where the definition of the group $G$  of TCTs is the set of maps $f$ such that the diagram commutes and $P\circ f$ is a diffeomorphism.

\subsection{Examples}

The simplest example is the case $n=1$, where there is only one force, $X$.  Now $\mathbb{RP}^{n-1}$ is just a point.   The group of diffeomorphisms of the point is a trivial group, consisting of only the identity element.  Any bundle over a point is trivial, so in this case the total space of the bundle is just $\mathbb{R}^\times$ itself and so the group of TCTs is the group of maps from the point to $\mathbb{R}^\times$ which is just $\mathbb{R}^\times$ itself, the multiplicative group of nonvanishing real numbers $\alpha$.  The action of this group on the force $X$ is just multiplication by $\alpha$.  So there is a one to one correspondence between TCTs and nonzero real numbers $\alpha$.  Therefore we find that if there is only 1 thermodynamic force, then the TCTs are linear.

\noindent The case $n=2$ shows the full structure of the group.  The projective space $\mathbb{RP}^1$ is a semicircle with both extremes identified, which topologically is just the circle $S^1$.  Therefore the group of TCTs is locally the product of the group of diffeomorphisms of the circle, which physically describe the mixing between $X^1$ and $X^2$, with the group of scalings $S^1\rightarrow\mathbb{R}^\times$.  Now a rotation of the $(X^1,X^2)$ plane by 180 degrees is a rotation of $\mathbb{RP}^1$ all the way around and so it acts trivially on $\mathbb{RP}^1$.  However it corresponds to the element $-1$ of the maps from $\mathbb{RP}^1$ to $\mathbb{R}^\times$. So indeed the group $G$ is not simply a product of the groups of scalings and rotations, the scalings are nontrivially fibered over the rotations.  

\section{Algebraic Description of the TCT-group}\label{algebra}

In this section we shall provide the algebraic description of the TCT-group. In particular, we shall define the TCT-group and we enunciate the theorem satisfied by the TCT- group. Details related to the demonstration of this theorem can be found in the Appendix.

\noindent Let $S^{n-1}$ the $n-1$ dimensional unit sphere ($\left\Vert \mathbf{x}\right\Vert =1$), represented as a $C^\infty$ differentiable manifold, as a submanifold embedded in $\mathbb{R}^{n}$. Define the equivalent relation $\mathbb{R}$ as follows : $\mathbf{x}$, $ \mathbf{y}\in S^{n-1}$ are equivalent iff $\mathbf{y=\pm x}$. Denote by $\Gamma^p_n$ the subgroup of $Diff$ ($\mathbb{RP}^{n-1}$) and let $\mathbf{Y}\in\Gamma^p_n$ iff $\mathbf{Y}(-\mathbf{x})=-\mathbf{Y}(\mathbf{x})$ where $S^{n-1}\ni\mathbf{x}\rightarrow \mathbf{Y}(\mathbf{x}) \in S^{n-1}$. The TCT-group, denoted by $G^s$, is the subgroup of homogeneous diffeomorphisms from $Diff(\mathbb{R}^n\backslash\{0\})$ i.e., $Diff(\mathbb{R}^n\backslash\{0\})\ni\mathbf{x\mapsto Y}_{g}(\mathbf{x}
)\in Diff(\mathbb{R}^{n}\backslash\{0\})$. Then, $\mathbf{Y}_{g}\in G^n$ iff
\begin{equation}\label{algebra1}
\mathbf{Y}_{g}(\lambda\mathbf{x}) =\lambda\mathbf{Y}_{g}(\mathbf{x}
);~\lambda\in\mathbb{R},~g\in G^{n}
\end{equation}
\noindent It is possible to demonstrate that the TCT-group $G^n$ may be split in a semidirect product of two subgroups where the first one is a Normal, abelian, subgroup. In particular, let us introduce two subgroups $N^n$ and $H^n$ defined as follows.

\noindent Let $N^{n}$ denote the subset (normal subgroup) of $G^{n}$ having the form
\begin{equation}\label{algebra2}
\mathbf{Y}_{g_{{}}}(\mathbf{x})=\mathbf{x~}r_{g}(\mathbf{x})~;~g\in
N^{n}\subset G^{n}
\end{equation}
\noindent with $r_{g}(\mathbf{x})$ denoting a positive $C^{\infty}(\mathbb{R}^{n}
\backslash\{0\})$ homogeneous function, i.e., 
\begin{equation}\label{algebra3}
r_{g}(\lambda\mathbf{x})=r_{g}(\mathbf{x})>0;~\lambda\in\mathbb{R}
\end{equation}
\noindent Denote by $H^{n}$ the subgroup of $G^{n}$ with the properties
\begin{equation}\label{algebra4}
\left\Vert \mathbf{Y}_{h}(\mathbf{x})\right\Vert  =\left\Vert
\mathbf{x}\right\Vert \quad ; \quad
\mathbf{Y}_{h}(-\mathbf{x})  =-\mathbf{Y}_{h}(\mathbf{x})\quad {\rm with} \quad
h \in H^{n}
\end{equation}
\noindent As it is proved in the Appendix, {\it the TCT-group is the semidirect product of the abelian normal subgroup $N^n$ and the subgroup} $H^n$, i.e., 
\begin{equation}\label{algebra5}
G^{n}=N^{n}~\rtimes H^{n}
\end{equation}
\noindent The irreducible representations of the group $G$ are then related to the irreducible representations of the subgroups $H$ and $N$. In the previous section we have shown that $G$ is a bundle whose base is $Diff$ ($\mathbb{RP}^{n-1}$). Expression~(\ref{algebra5}) specifies, in more rigorous terms, which bundle it is.

\subsection{Properties of the General Element of the TCT-Group}
From Eq.~(\ref{topology1}), we easily get (no Einstein's convention on the repeated indexes)
\begin{align}\label{algebra6}
&U^\mu_\nu\equiv\frac{\partial {X^{}}^{'\mu}}{\partial {X^{}}^\nu}={F^{}}^\mu\delta^1_{\nu}+\frac{\partial {F^{}}^\mu}{\partial {Y^{}}^\nu}\frac{X^1}{{X^{}}^{\nu-1}}(1-\delta^1_\nu)-\frac{{X^{}}^1{X^{}}^{\nu+1}}{({{X^{}}^\nu})^2}\frac{\partial {F^{}}^\mu}{\partial {Y^{}}^{\nu+1}}(1-\delta^n_\nu)\quad {\rm with} \nonumber\\
&Y^\nu\equiv\frac{X^\nu}{X^{\nu-1}}\qquad ;\qquad \mu ,\ \nu=1\cdots, n
\end{align}
\noindent As shown in \cite{sonnino}, matrix $U^\mu_\nu$ satisfies the important relations 
\begin{equation}\label{algebra7}
X^\nu\frac{\partial U_\nu^\mu}{\partial X^\kappa}=0\qquad ; \qquad X^\kappa\frac{\partial U_\nu^\mu}{\partial X^\kappa}=0
\end{equation}
\noindent Close to the identity, it is useful to write the TCT as 
\begin{equation}\label{algebra8}
\qquad
\left\{ \begin{array}{ll}
U^\mu_\nu=\delta^\mu_\nu+\epsilon^\alpha\delta U^\mu_{\nu(\alpha)} & \ \ \mbox{{\rm with}}\\
X^{'\mu}=X^\mu+\epsilon^\alpha\xi^\mu_{(\alpha)} & \ \ \mbox{}
\end{array}
\right.
\quad
\left\{ \begin{array}{ll}
\delta U^\mu_{\nu(\alpha)}=\omega^\mu_{\nu(\alpha)}+X^\kappa\partial_\nu\omega_{\kappa (\alpha)}^\mu  \\
\xi^\mu_{(\alpha)}=\omega^\mu_{\nu(\alpha)}X^\nu 
\end{array}
\right.
\end{equation}
\noindent where $\epsilon^\alpha$ are infinitesimal parameter coefficients.

\subsection{Examples}

\noindent {\bf Linear TCT}

\noindent In general, the TCT-group, $G^n$, is a non-compact, infinite Lie-group. However, $G^n$ admits several compact and finite subgroups. Linear transformations of the thermodynamic forces are an important subgroup of the $G^n$. A significant example is the two-dimensional linear transformations
\begin{equation}\label{algebra8}
\left\{ \begin{array}{ll}
X^{'1} =a_1X^1+\epsilon_1 a_2X^2=X^1+\epsilon_1\xi^1& \ \ \mbox{}\\
X^{'2} =\epsilon_2b_1X^1+ b_2X^2=X^2+\epsilon_2\xi^2 & \ \ \mbox{}
\end{array}
\right.
\qquad {\rm with}\qquad
\left\{ \begin{array}{ll}
a_1-1=\epsilon_1\alpha_1 \\
b_2-1=\epsilon_2\beta_2
\end{array}
\right.
\end{equation}
\noindent and
\begin{equation}\label{algebra9}
\omega^\mu_\nu=
\begin{pmatrix}
\alpha_1 & a_2\\
b_1 & \beta_2
\end{pmatrix}
\qquad {\rm and} \qquad
\xi^\mu=
\begin{pmatrix}
\alpha_1X^1+ a_2X^2\\
b_1X^1+\beta_2X^2
\end{pmatrix}
\end{equation}
\noindent The four generators of the group are
\begin{equation}\label{algebra10}
t_1=-iX^1\partial_{X^1}\quad ;\quad t_2=-iX^1\partial_{X^2}\quad ; \quad t_3=-iX^2\partial_{X^1}\quad ; \quad t_4=-iX^2\partial_{X^2}
\end{equation}
\noindent The Lie algebra reads
\begin{align}\label{algebra11}
&\big[ t_\mu ,t_\mu\big]=0\quad ; \quad [t_\mu, t_\nu]=-[t_\nu,t_\mu]\quad ; \quad [t_2, t_4]=-it_2\quad ; \quad  [t_3, t_4]=it_3\nonumber \\
& [t_1, t_2]=-it_2\quad ; \quad  [t_1, t_3]=it_3\quad ; \quad [t_1, t_4]=0\quad ; \quad [t_2, t_3]=it_4-it_1
\end{align}
\noindent From the Lie-algebra, we may construct the adjoint representations of the generators of the group, $T^{(\kappa)}_{\mu\nu}$, through the structure constants
\begin{equation}\label{algebra12}
T^{(\kappa)}_{\mu\nu}=if^\nu_{\mu(\kappa)}\qquad {\rm with}\qquad [t_\mu, t_\kappa]=f^\nu_{\mu(\kappa)}t_\nu
\end{equation}
\noindent We get
\begin{align}\label{algebra13}
&T^{(1)}_{\mu\nu}=
\begin{pmatrix}
0 & 0 & 0 & 0\\
0 & i & 0 & 0\\
0 & 0 & -i & 0\\
0 & 0 & 0 & 0
\end{pmatrix}
\qquad {\rm ;} \qquad
T^{(2)}_{\mu\nu}=
\begin{pmatrix}
0 & -i & 0 & 0\\
0 & 0 & 0 & 0\\
i & 0 & 0 & -i\\
0 & i & 0 & 0
\end{pmatrix}\nonumber \\
&T^{(3)}_{\mu\nu}=
\begin{pmatrix}
0 & 0 & i & 0\\
-i & 0 & 0 & i\\
0 & 0 & 0 & 0\\
0 & 0 & -i & 0
\end{pmatrix}
\qquad {\rm ;} \qquad
T^{(4)}_{\mu\nu}=
\begin{pmatrix}
0 & 0 & 0 & 0\\
0 & -i & 0 & 0\\
0 & 0 & i & 0\\
0 & 0 & 0 & 0
\end{pmatrix}
\end{align}
\noindent It is worth mentioning that the previous transformations play an important role in physics, for example in Tokamak-plasmas in fully collisional transport regime (the so-called {\it Pfirsch-Schl$\ddot u$ter transport regime}) \cite{balescu2} and \cite{sonnino3}. Here the two thermodynamic forces read $X^1=(n_eT_e)^{-1}\nabla_rP$ and $X^2=-T_e^{-1}\nabla_rT_e$, with $n_e$, $T_e$ and $P$ denoting the electron density number, the electron temperature and the total pressure of the plasma, respectively. In this case, the subgroup of $G^n$ is finite and compact. 

\noindent Another example of linear TCT, widely used in Tokamak-plasmas in the weak-collisional transport regime (the so called {\it banana regime}), is provided by the Hinton-Hazeltine transformations \cite{hinton}. In this case, the TCT read
\begin{equation}\label{algebra9a}
\left\{ \begin{array}{ll}
X^{'1}=X^1-\frac{5}{2}X^2-\frac{5}{2}Z^{-1}X^3 \\
X^{'2}=X^2\\
X^{'3}=X^3\\
X^{'4}=X^4\\
\end{array}
\right.
\end{equation}
\noindent with $Z$ denoting the charge number. In this particular case, $X^1=-(n_eT_e)^{-1}\nabla_rP$, $X^2=-T_e^{-1}\nabla_rT_e$, $X^3=-T_i^{-1}\nabla_rT_i$, and $X^4=<B^2>^{-1/2}<B E_\parallel^A>$, with $B$ and $E^A_\parallel$ denoting the {\it intensity of the magnetic field} and the {\it electric field generated by the external coils, parallel to the magnetic field,} respectively. The angular brackets denotes the {\it averaged magnetic surface operation} (see, for example, \cite{balescu2}). In this case the space of the thermodynamic forces is four-dimensional. The TCT-group possesses 16 generators, which are similar to the ones given by Eqs~(\ref{algebra10}), with adjoint representations also similar to Eqs~(\ref{algebra13}), but the dimension of the matrices are $16\times 16$. Note that also in this case, the TCT-subgroup is compact and finite.
\vskip 0.5truecm 
\noindent {\bf Nonlinear TCT}

\noindent The linear transformations are an example of (closed) sub-algebra of the TCT-group. However, it is easy to convince ourselves that nonlinear examples of TCT sub-algebras may also be found. Consider, for example the following TCT
\begin{equation}\label{algebra14}
\left\{ \begin{array}{ll}
X^{'1} =a_1X^1& \ \ \mbox{}\\
X^{'2} =b_1X^1+ b_2\Bigl(\frac{X^2}{X^1}\Bigr)X^2& \ \ \mbox{}
\end{array}
\right.
\qquad {\rm with}\qquad
\left\{ \begin{array}{ll}
a_1-1=\epsilon_1\alpha_1 \\
b_2-1=\epsilon_2\beta_2
\end{array}
\right.
\end{equation}
\noindent The three generators of the group and their adjoint representations read, respectively
\begin{align}\label{algebra15}
&t_1=-iX^1\partial_{X^1}\quad ;\quad t_2=-iX^2\partial_{X^2}\quad ; \quad t_3=-i\Bigl(\frac{X^2}{X^1}\Bigr)X^2\partial_{X^2} \nonumber \\
&T^{(1)}_{\mu\nu}=
\begin{pmatrix}
0 & 0 & 0 \\
0 & 0 & 0 \\
0 & 0 & -i 
\end{pmatrix}
\quad {\rm ;} \quad
T^{(2)}_{\mu\nu}=
\begin{pmatrix}
0 & 0 & 0 \\
0 & 0 & 0 \\
0 & 0 & i 
\end{pmatrix}
\quad {\rm ;} \quad
T^{(3)}_{\mu\nu}=
\begin{pmatrix}
0 & 0 & i \\
0 & 0 & -i \\
0 & 0 & 0 
\end{pmatrix}
\end{align}

\section{Noether's Current for Fully Collisional Tokamak-plasmas}\label{noether}

As known, through Noether's theorem one can determine the conserved quantities from the observed symmetries of a physical system. In particular, consider the action 
\begin{equation}\label{n1}
I=\int\mathcal{L}(\Phi^A,\partial_\mu\Phi^A, X^\mu)\sqrt{g}dX^n
\end{equation}
\noindent with $\phi^A$ denoting the set of differentiable fields defined over all space of the thermodynamic forces and $\mathcal{L}$ the Lagrangian density, respectively. In our case $\Phi^A=\{g_{\kappa\nu}, f_{\kappa\nu},\Gamma^\lambda_{\kappa\nu}\}$. Let the action be invariant under certain transformations of the thermodynamic forces coordinates $X^\mu$ and the field $\Phi^A$
\begin{equation}\label{n2}
\left\{ \begin{array}{ll}
X^\mu\rightarrow X^\mu+\delta X^\mu=X^\mu+\epsilon_\alpha\xi^\mu_{(\alpha)} \\
\Phi^A(X)\rightarrow\Phi^A(X)+\delta\Phi^A(X)=\Phi^A(X)+{\bar\delta}\Phi^A(X)+{\tilde\delta}\Phi^A(X)=\Phi^A(X)+\epsilon_\alpha\Psi^A_{(\alpha)}(X)
\end{array}
\right.
\end{equation}
\noindent where $\delta\Phi^A$ denotes the {\it transformation in the field variables}, ${\bar\delta}\Phi^A$ the {\it intrinsic changes of the field}, and ${\tilde\delta}\Phi^A$ the {\it transformation of the field variables due to the coordinates variation}, respectively. Noether's theorem states that $N$ {\it currents densities} are conserved, with $N$ equals to the number of generators of the Lie group associated to the TCT \cite{noether}, \cite{noether1}. In our case, the action remains invariant only under TCT (and not under the field transformations). Hence, the expressions of the $N$ Noether currents $j^\mu_\alpha$ reduce to
\begin{eqnarray}\label{n3}
&&j^\mu_\alpha=\frac{\partial L}{\partial\Phi^A_{,\mu}}{\mathcal L}_{{\mathbf\xi}_\alpha}\Phi^A-L\xi^\mu_\alpha\qquad{\rm with} \\
&&\partial_\mu\bigl[\sqrt{g}J^\mu_\alpha\bigr]=0\quad;\quad J^\mu_\alpha\equiv g^{-1/2} j^\mu_\alpha \quad;\quad(\alpha=1,\cdots ,N)\nonumber
\end{eqnarray}
\noindent Here, ${\mathcal L}_\xi$ denotes the Lie derivatives along the $\xi^\mu_\alpha$ vector and $L\equiv{\mathcal L}\sqrt{g}$, respectively.

\noindent As an example of application, let us consider the action (\ref{pa5}) and the case of Tokamak-plasmas in fully collisional transport regime, with the TCT given by Eqs~(\ref{algebra8})-(\ref{algebra11}). After (some tedious) calculations, and under the realistic approximation $1/\sigma\ll 1$ valid for Tokamak-plasmas, we finally get (see Appendix)
\begin{equation}\label{n4}
J^{\mu\lambda}_\nu=\frac{1}{2}g^{\kappa\lambda}A^{\mu\eta}_{\nu\kappa\eta}+\frac{1}{2}g^{\beta\lambda}A^{\mu\eta}_{\nu\eta\beta}-g^{\kappa\beta}A^{\mu\lambda}_{\nu\kappa\beta}\ \ {\rm ;}\ \ A^{\mu\eta}_{\nu\kappa\beta}=X^\mu\Gamma^\eta_{\kappa\beta, \nu}-\Gamma^\mu_{\kappa\beta}\delta^\alpha_\nu+\Gamma^\alpha_{\nu\beta}\delta^\mu_\kappa+\Gamma^\mu_{\nu\kappa}\delta_\beta^\mu
\end{equation}
\noindent Note that there are no Noether's currents in the fully collisional transport regimes since in this case  all derivatives of the transport coefficients with respect to $X^\mu$ (and hence, $\Gamma^\kappa_{\mu\nu}$ and its derivatives $\Gamma^\kappa_{\mu\nu,\eta}$) are identically equal to zero. These currents appear only in the nonlinear transport regime where the derivatives of the transport coefficients with respect to the thermodynamic forces do not vanish \cite{sonnino}. Even though calculations are more complex, it is possible to show that the above conclusions apply also to the weak collisional (banana) and the plateau transport regimes. Fig.~\ref{noether1} shows one component of Noether's current, $J^{11}_1$, against the two thermodynamic forces $X^1$ and $X^2$. The contour plot of this current is illustrated in Fig.~(\ref{noether2}) [these graphics have been produced by Philippe Peeters, from the Universit{\'e} Libre de Bruxelles (ULB) - Brussels (Belgium)]. 

%%%%%%%%%%%%%%%%%%%%%%%%%%%%%%%%%%%%%%%%%%%%%%%%
\begin{figure*}
\hfill 
\begin{minipage}[t]{.45\textwidth}
    \begin{center}  
\hspace{-1.2cm}
\resizebox{1\textwidth}{!}{%
\includegraphics{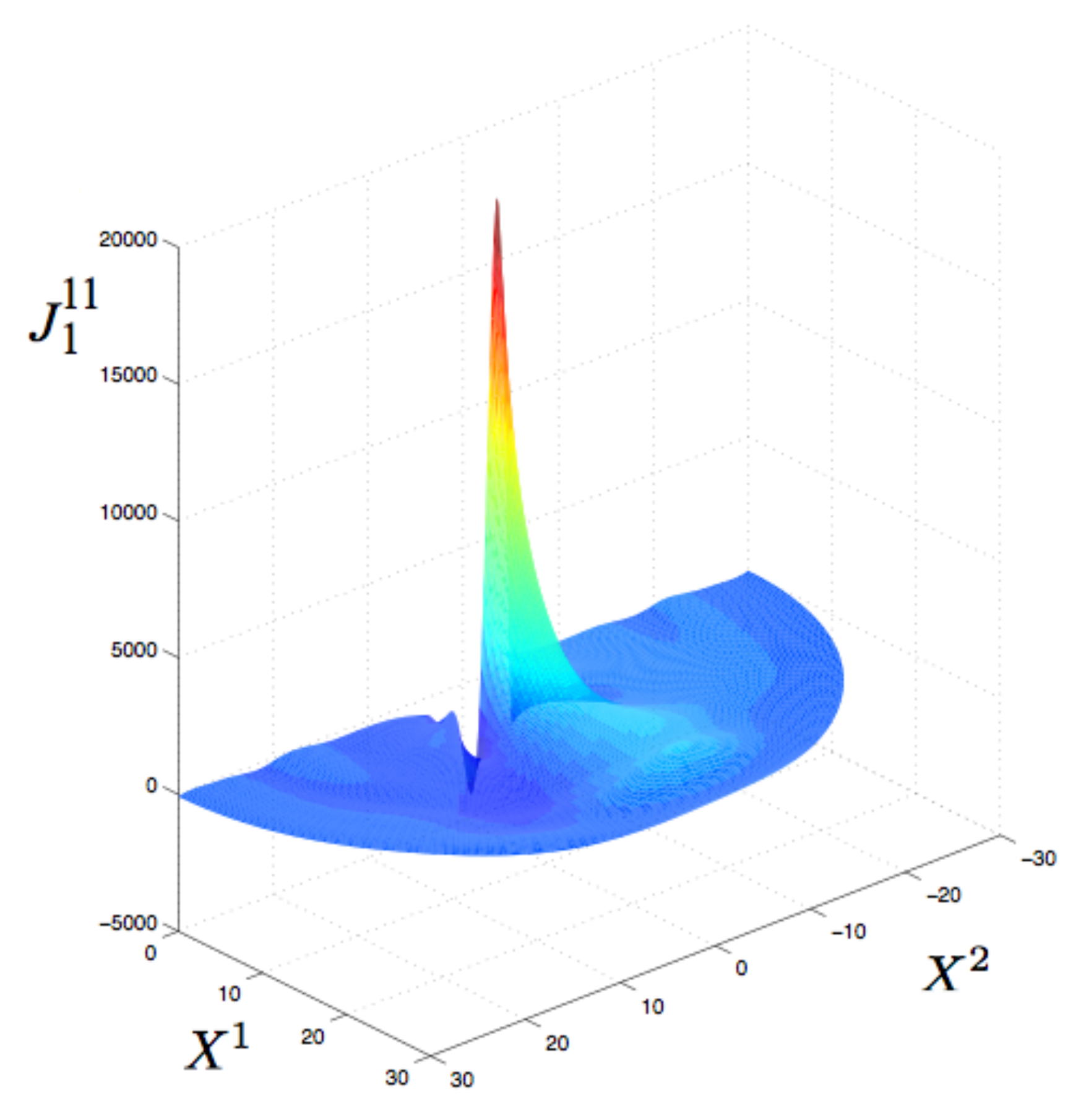}
}
\caption{ \label{noether1}Component $J^{11}_1$ corresponding to Eq.~(\ref{n3}).}
\end{center}
  \end{minipage}
\hfill 
\begin{minipage}[t]{0.5\textwidth}
    \begin{center}
\hspace{-0.9cm}
\resizebox{1\textwidth}{!}{%
\includegraphics{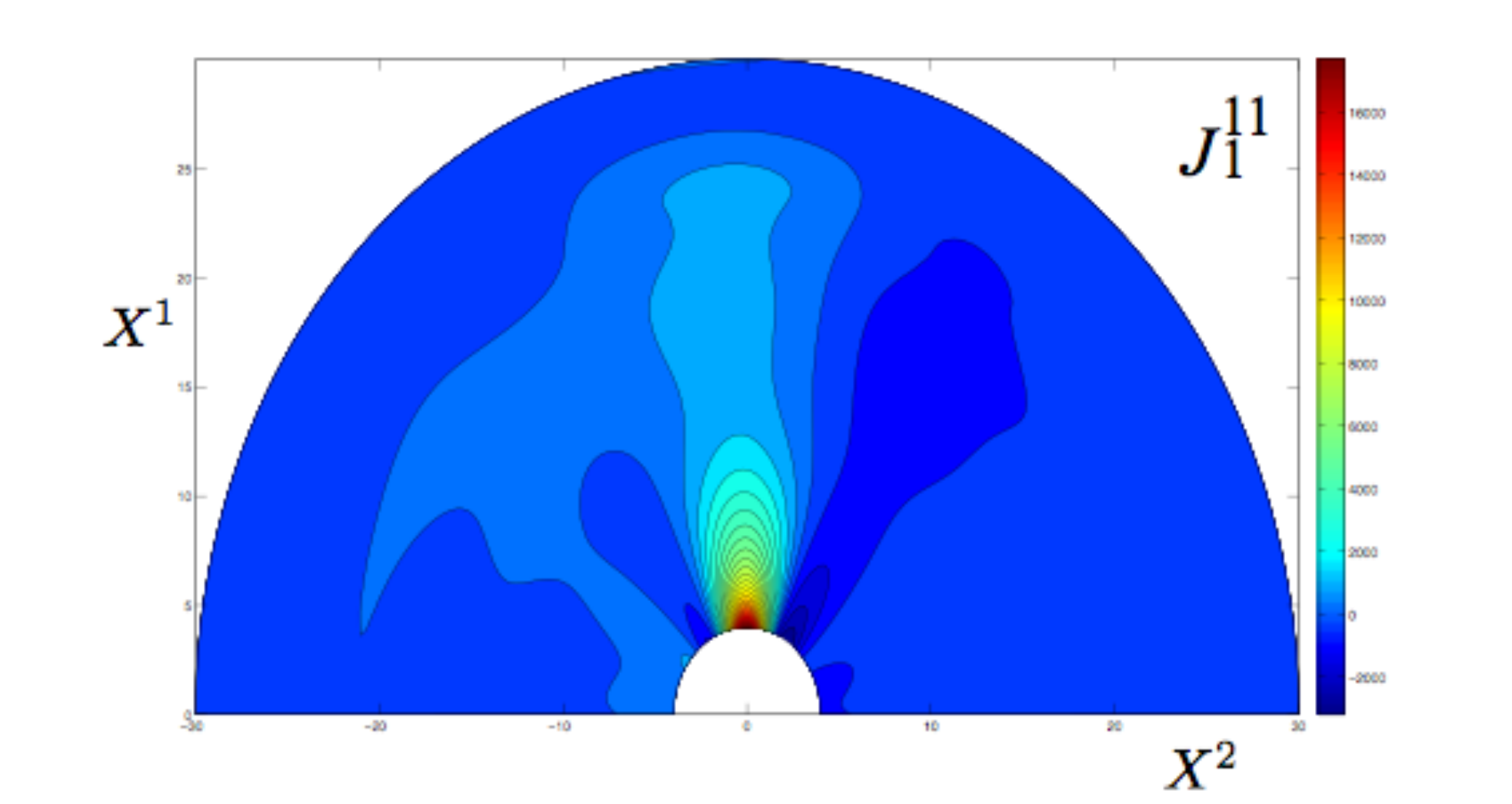}
}
\caption{Contour plot of the $J^{11}_1$ profile.}
\label{noether2}
\end{center}
  \end{minipage}
\hfill
\end{figure*}
%%%%%%%%%%%%%%%%%%%%%%%%%%%%%%%%%%%%%%%%%%%%%%%%

\section{Derivation of the Collisional Operator that ensures, at the lowest order, the Covariance under TCT of the Closure Transport Relations}\label{kinetic}

The aim of this section is to derive the expression of the collisional operator for magnetically confined plasmas, which guarantees that the Thermodynamic Covariance Principle (TCP) is satisfied by the closure transport relations (i.e., the flux-force relations). Let us consider a two-component system of charged particles. The statistical state is represented by two reduced distribution functions $f^\alpha$ corresponding to ions $i$ and electrons $e$ \cite{balescu1} (no Einstein's convention on index $\alpha$)   
\begin{equation}\label{c1}
\frac{\partial}{\partial t}f^\alpha ({\bf q},{\bf v},t)=-{\bf v}\cdot\frac{\partial}{\partial{\bf q}}f^\alpha({\bf q},\bf{v},t)-\frac{e_\alpha}{m_\alpha}\Bigl({\bf E}({\bf q},t)+\frac{1}{c}{\bf v}\wedge{\bf B}({\bf q})\Bigr)\cdot\frac{\partial}{\partial{\bf v}}f^\alpha({\bf q},\bf{v},t)+C^\alpha(f,f)
\end{equation}
\noindent Here $\alpha =e,i$ and $c$ is the speed of light in vacuum. Moreover, $e_\alpha$, and $m_\alpha$ are the {\it charge} and the {\it mass} of species $\alpha$, ${\bf q}$ and {\bf v} denote the {\it generalized coordinates} and the {\it velocity of the particle}, and ${\bf E}$ and ${\bf B}$ the {\it electric} and the {\it magnetic fields}, respectively. Note that the first term on the r.h.s. of Eq.~(\ref{c1}) represents the {\it free flow}, the second term corresponds to the {\it electromagnetic contribution}, and the last term is the contribution due to {\it collisions}. In many applications in plasmas physics (including those involving the radio-frequency waves) collisions dominate the thermal particles. Therefore, the distribution function can conveniently be expanded about a Maxwellian 
\begin{equation}\label{c2}
f^\alpha ({\bf x},{\bf v},t)=f^\alpha_{eq.}({\bf x},{\bf v},t)[1+\chi({\bf x},{\bf v},t)]
\end{equation}
\noindent where {\bf x} denotes the {\it position of the particle}. In Eq.~(\ref{c2}) we have introduced the {\it reference state} $f^\alpha_{eq.}({\bf x},{\bf v},t)$, i.e. the local plasma equilibrium (L.P.E.), and the {\it deviation from the reference state} $\chi$. The local plasma equilibrium is defined in the following way. The electron-electron and ion-ion collisions bring the plasma in a short time to a state of local plasma equilibrium satisfying the equations 
\begin{equation}\label{c2a}
C^{ee}=C^{ii}=0
\end{equation}
\noindent with $C^{ee}$ and $C^{ii}$ denoting the electron-electron and ion-ion collisions, respectively (see below the expression of the collisional operator). The L.P.E. is the solution of Eqs~(\ref{c2a}):
\begin{equation}\label{c3}
f^\alpha_{eq.}({\bf x},{\bf v},t)=(2\pi)^{-3/2}n_\alpha({\bf x},t)\Bigl(\frac{m_\alpha}{T_\alpha({\bf x},t)}\Bigl)^{3/2}\!\!\!\!\exp(-{\bf c}_N\cdot{\bf c}_N)\ ;\ {\bf c}_N\equiv\Bigl(\frac{m_\alpha}{T_\alpha}\Bigr)^{1/2}[{\bf v}-{\bf u}({\bf x},t)]
\end{equation}
\noindent where ${\bf u}^\alpha$, $n_\alpha$, $T_\alpha$ are the {\it mean velocity}, the {\it number density} and {\it temperature}, respectively. The deviation $\chi$ may be developed in terms of the Hermite polynomials $H^{(m)}_{r_1r_2\cdots}$
\begin{eqnarray}\label{c4}
\chi^\alpha({\bf x},{\bf c}_N, t)=&&\sum_{n=0}^\infty q^{\alpha(2n)}({\bf x},t)H^{(2n)}({\bf c}_N)+\sum_{n=0}^\infty q_r^{\alpha(2n+1)}({\bf x},t)H_r^{(2n+1)}({\bf c}_N)+\nonumber \\
&&\sum_{n=0}^\infty q_{rs}^{\alpha(2n)}({\bf x},t)H_{rs}^{(2n)}({\bf c}_N)+\cdots
\end{eqnarray}
\noindent with $q^{\alpha(m)}({\bf x},t)$ denoting the hermitian moments \cite{balescu1}. The Landau collisional operator can be brought into the form \cite{balescu2}, \cite{landau}
\begin{align}\label{c5}
&C^\alpha(f,f)=\sum_{\beta=e,i}C^{\alpha\beta}\qquad {\rm with}\\
&C^{\alpha\beta}(f^\alpha(1),f^\beta(2))=2\pi e^2_\alpha e^2_\beta\log\Lambda\int d{\bf v}_2{\tilde\partial}_r G_{rs}({\bf v}_1-{\bf v}_2){\tilde\partial}_sf^\alpha (1)f^\beta(2)\nonumber 
\end{align}
\noindent with $r, s$ identifying the components of a vector, and indexes $(1)$ and $(2)$ the colliding particles $(1)$ and $(2)$, i.e. $(1)\equiv ({\bf x}_1, {\bf v}_1, t)$ and $(2)\equiv ({\bf x}_2, {\bf v}_2, t)$. Here, $\log\Lambda$ and $G_{rs}$ are the {\it Coulomb logarithm} (linked to the {\it Debye length} $\lambda_D$) and the {\it Landau tensor}, respectively, i.e.
\begin{equation}\label{c6}
\log\Lambda=\ln\frac{3(T_e+T_i)\lambda_D}{2Ze^2}\ \ ; \ \  \lambda_D=\Bigl(\frac{4\pi Ze^2(n_eT_e+n_iT_i}{T_eT_i(1+Z)}\Bigr)^{-1/2}\ \ ;\ \ G_{rs}({\bf a})=\frac{a^2\delta_{rs}-a_ra_s}{a^3}
\end{equation}
\noindent with ${\bf a}$ denoting the {\it relative velocity of two particles} i.e., ${\bf a}\equiv {\bf v}_1-{\bf v}_2$. The operator ${\tilde\partial}_r$ is defined as follows
\begin{equation}\label{c7}
{\tilde\partial}_r\equiv m_\alpha^{-1}\partial_{v_{1r}}-m_\beta^{-1}\partial_{v_{2r}}
\end{equation}
\noindent By inserting Eqs~(\ref{c2})-(\ref{c4}) into Eq.~(\ref{c1}), and by truncating the expansion up to the second order of the (small) {\it drift parameter} $\epsilon$ (defined as the Larmour radius over a macroscopic length), we get the {\it vector moment equations} \cite{balescu2}
\begin{eqnarray}\label{c8}
&&\Omega_{\alpha}\tau_{\alpha}\epsilon_{rmn}q_m^{\alpha (1)}b_n+\tau_{\alpha}Q_r^{\alpha (1)}+g_r^{\alpha (1)}+{\bar g}_r^{\alpha (1)}+O(\epsilon^2)=0\nonumber \\
&&\Omega_{\alpha}\tau_{\alpha}\epsilon_{rmn}q_m^{\alpha (3)}b_n+\tau_{\alpha}Q_r^{\alpha (3)}+g_r^{\alpha (3)}+{\bar g}_r^{\alpha (3)}+O(\epsilon^2)=0\nonumber \\
&&\Omega_{\alpha}\tau_{\alpha}\epsilon_{rmn}q_m^{\alpha (5)}b_n+\tau_{\alpha}Q_r^{\alpha (5)}+{\bar g}_r^{\alpha (5)}+O(\epsilon^2)=0\nonumber \\
&& Q_r^{\alpha(m)}=n_\alpha^{-1}\int d{\bf v} H_r^{(m)}\Bigl((m_\alpha/T_\alpha)^{1/2}({\bf v}-{\bf u}^\alpha)\Bigr)C^\alpha
\end{eqnarray}
\noindent where the Einstein convention is adopted on the repeated indexes $m$ and $n$, but not on the index $\alpha$. Here, $b_n$ is a unit vector along the magnetic field {\bf B} i.e., $b_n\equiv B_n/B$, $\epsilon_{rmn}$ is the completely antisymmetric Levi-Civita symbol, $\Omega_{\alpha}$ is the Larmor frequency of species $\alpha$ and $\tau_\alpha$ is the {\it relaxation time of species} $\alpha$, respectively. Moreover, $g_r^{\alpha (n)}$, ${\bar g}_r^{\alpha(n)}$, and $Q_r^{\alpha (n)}$ are the {\it dimensionless source terms related to the thermodynamic forces}, the {\it additional sources terms in the long mean free path transport regime}, and the {\it dimensionless friction terms}, respectively (the exact definitions of these quantities may be found in Ref.~\cite{balescu2}).

\noindent For collision-dominated plasmas (i.e., in absence of turbulence), the entropy production, $\Sigma^\alpha$, of the plasma for species $\alpha$ may be brought into the form \cite{balescu2}
\begin{equation}\label{c9}
\Sigma^\alpha=-\tau_\alpha\sum^N_{n=0(1)}q_r^{\alpha(n+1)}Q_r^{\alpha(2n+1)}
\end{equation}
\noindent The lower limit for $n$ is $0$ for the electrons and $1$ for the ions. Hence, thanks to this theorem, $Q_r^{\alpha(n)}$ and $q_r^{\alpha (n)}$ are the {\it thermodynamic forces} and the {\it thermodynamic fluxes} for magnetic confined plasmas, respectively. Eq.~(\ref{c9}) tells us that the last equation of Eqs~(\ref{c8}) is the closure equation (flux-forces relation) for Tokamak-plasmas, derived by kinetic theory. The region where the transport coefficients do not depend on the thermodynamic forces is referred to as {\it Onsager's region} or, the {\it linear thermodynamic  regime}. A well-founded microscopic explanation on the validity of the linear phenomenological laws was developed by Onsager in 1931 \cite{onsager}-\cite{onsager1}. Onsager's theory is based on three assumptions: i) {\it The probability distribution function for the fluctuations of thermodynamic quantities} (Temperature, pressure, degree of advancement of a chemical reaction etc.) {\it is a Maxwellian} ii) {\it Fluctuations decay according to a linear law} and iii) {\it The principle of the detailed balance} (or the microscopic reversibility) {\it is satisfied}. Out of Onsager's regime, the transport coefficients may depend on the thermodynamic forces. This happens when the above-mentioned assumption 1) end/or assumption 2) are/is not satisfied. Magnetically confined tokamak plasmas are a typical example of thermodynamic systems out of Onsager's region. In this case, even in absence of turbulence, the local distribution functions of species (electrons and ions) deviate from the (local) Maxwellian [see Eq.~(\ref{c2})]. After a short transition time, the plasma remains close to (but, it is not in) a state of local equilibrium (see, for example, \cite{sonnino3}, \cite{balescu2}). The neoclassical theory is a linear transport theory (see, for example, \cite{balescu2}) meaning by this, a theory where the moment equations are coupled to the closure relations (i.e. flux-force relations), which have been linearized with respect to the generalized frictions (see, for example, Ref.~\cite{balescu1}). This approximation is clearly in contrast with the fact that the distribution function of the thermodynamic fluctuations is {\it not} a Maxwellian and it could be a possible cause of disagreement between the theoretical predictions and the experimental profiles \cite{sonnino3}, \cite{sonnino4}. However, it is important to mention that it is well accepted that the main reason of this discrepancy is attributed to turbulent phenomena existing in Tokamak-plasmas. Fluctuations in plasmas can become unstable and therefore amplified, with their nonlinear interaction, successively leading the plasma to a state, which is far away from equilibrium. In this condition, the transport properties are supposed to change significantly and to exhibit qualitative features and properties that could not be explained by collisional transport processes, e.g. size-scaling with machine dimensions and non-local behaviors that clearly point at turbulence spreading etc. (see, for example, Ref.~\cite{diamond}). Hence, the truly complete transport theory of plasmas must self-consistently incorporate the instability theory that includes the influence of nonlinear transformations on fluctuations. This global approach is the purpose of the so-called {\it anomalous transport theory} (still far from a complete and comprehensive theory). This type of problem is, however, far beyond the scope of the present work. Here, more modestly, we deal with plasmas in the collisional-dominated transport regime, characterized by a time-scale which is much longer than one involved in the so-called {\it fluctuation-induced turbulence transport}.

\noindent Our aim is to determine the simplest expression of the collisional operator such that the resulting closure equation satisfies the TCP (without, of course, violating the energy, mass, and momentum conservation laws). Concretely, in mathematical terms, we need to identify an operator able {\it to kill} the terms that do not satisfy the TCP and, in order not to violate the conservation laws, which commutes with the operator ${\tilde\partial}_r$. The last equation in Eq.~(\ref{c8}) will satisfy the TCP if 
\begin{equation}\label{c10}
{\rm when}\qquad C^{\alpha\beta}\rightarrow \lambda C^{\alpha\beta}\qquad {\rm then}\qquad Q_r^{\alpha (m)}\rightarrow \lambda Q_r^{\alpha (m)}
\end{equation}
\noindent with $\lambda$ denoting a constant parameter. We introduce now the operator $\mathcal{O}_{TCT}$ defined as follows
\begin{equation}\label{c11}
\mathcal{O}_{TCT}\equiv \Bigl(2-\chi^\alpha(1)\frac{\partial}{\partial\chi^\alpha(1)}-\chi^\beta(2)\frac{\partial}{\partial\chi^\beta(2)}\Bigr)
\end{equation}
\noindent It is easily checked that this operator possesses the following properties
\begin{equation}\label{c12}
\mathcal{O}_{TCT}(\chi^\alpha)=\chi^\alpha\ \; \ \ \mathcal{O}_{TCT}\bigl((\chi^\alpha)^2\bigr)=0\ \ ;\ \ \mathcal{O}_{TCT}(\chi^\alpha(1)\chi^\beta(2))=0\ \ ;\ \ \big[\mathcal{O}_{TCT},{\tilde\partial}_r]=0
\end{equation}
\noindent where the square brackets denote the {\it Lie brackets}. The last equation in Eq.~(\ref{c8}) (the closure equation), satisfies the TCP iff 
\begin{equation}\label{c13}
C^{\alpha\beta}_{TCT}=C^{\alpha\beta}\bigl({\mathcal O}_{TCT}(f(1)f(2)\bigr)
\end{equation}
\noindent or
\begin{equation}\label{c14}
C_{TCT}^{\alpha\beta}=2\pi e^2_\alpha e^2_\beta\log\Lambda\int d{\bf v}_2{\tilde\partial}_r G_{rs}({\bf v}_1-{\bf v}_2){\tilde\partial}_s\Bigl(f^\alpha (1)_{eq.}f^\beta(2)_{eq.}\bigl(\chi^\alpha(1)+\chi^\beta(2)\bigr)\Bigr)
\end{equation}
\noindent Eq.~(\ref{c14}) can conveniently be written in the form 
\begin{align}\label{c15}
&C_{TCT}^{\alpha\beta}=2\pi e^2_\alpha e^2_\beta\ln\Lambda\int d{\bf v}_2{\tilde\partial}_rG_{rs}(\!{\bf v}_1\!-\!{\bf v}_2\!){\tilde\partial}_s[f^\alpha(1)f^\beta_{eq.}(2)]_{\!_+}\quad {\rm with}\\
&[f^\alpha(1)f^\beta_{eq.}(2)]_{\!_+}\equiv f^\alpha(1)f_{eq.}^\beta(2)+f_{eq.}^\alpha(1)f^\beta(2)\nonumber
\end{align}
\noindent Thanks to the last relation in Eqs~(\ref{c12}), we also get  
\begin{align}\label{c16}
&\int d{\bf v}C_{TCT}^{\alpha}=0\qquad\qquad\qquad\  (\alpha=e,i)\qquad\qquad\! {\rm Number\ of\ particles\ conservation}\\
&\sum_\alpha m_\alpha \int d{\bf v}\ v_r C_{TCT}^{\alpha}=0\qquad (r=1,2,3)\qquad\quad\!  {\rm Momentum\ conservation}\\
&\sum_\alpha \frac{1}{2}m_\alpha \int d{\bf v}\ v^2 C_{TCT}^{\alpha}=0\qquad\qquad\qquad\qquad\qquad\!\!\!\!\!{\rm Energy\ conservation}
\end{align}
\noindent with $C_{TCT}^{\alpha}=\Sigma_{\beta=e,i}C_{TCT}^{\alpha\beta}$. Eq.~ (\ref{c15}) is the linearized collision operator used in existing literature \cite{hinton}, \cite{karney}. However, it should be noted that in previous literature the quadratic contributions in the distribution functions are ignored without any physical justification. Here, on the contrary, the linearization process of the collisional operator rests upon the validity of a fundamental principle, that is the Thermodynamic Covariance Principle (TCP). Notice that, to linearize the collisional operator does not mean that we are in the Onsager regime. As known, this regime is attained by performing two operations: 1) the transport phenomena is evaluated by determining a finite number of Hermitian moments of the distribution functions, and 2) the truncated set of moment equations is linearized in some appropriate way \cite{balescu2}, \cite{balescu1}.

\noindent In order to test the validity of the TFT, we have computed a concrete example of heat loss in L-mode, collisional, Tokamak-plasmas. To perform correctly this calculation, we have solved Eq.~(\ref{pa9}), subject to boundary conditions, by taking into account the so-called Shafranov's shift. The Shafranov shift is the outward radial displacement $\Delta (r)$ of the centre of the magnetic flux surfaces with the minor radius $r$ of the Tokamak, induced by the {\it plasma pressure} $\beta_{mag}$. This shift compresses the surfaces on the outboard side \cite{shafranov}. In terms of number density, Temperature and the intensity of the magnetic field, $\beta_{mag}=P/P_{mag}$, with $P=nT$ and $P_{mag}$ given in the forthcoming Eq.~(\ref{shafranov1}). Fig.~(\ref{Shafranov_shift}) depicts the Shafranov shift  due to $\beta_{mag}$, versus the minor radius of the Tokamak.
%%%%%%%%%%%%%%%%%%%%%%%%%%%%%%%%%%%%%%%%%%%%%%%%%%%%%
\begin{figure*}[htb] 
\hspace{-1.0cm}\includegraphics[width=15cm,height=6cm]{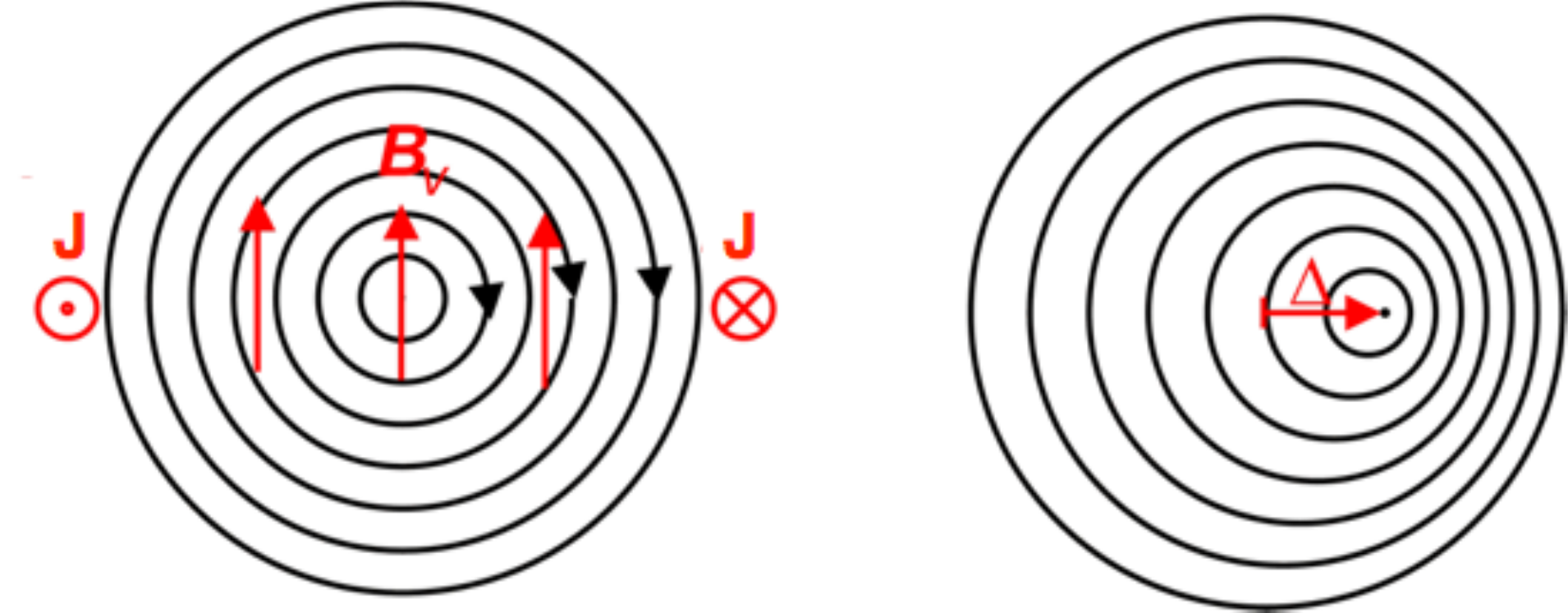}
\caption{ \label{Shafranov_shift} Shafranov shift. In Tokamak-plasmas, the plasma pressure leads to an outward shift, $\Delta$, of the centre of the magnetic flux surfaces. $J$ indicates the direction of the electric current that flows inside the plasma. Note that the poloidal magnetic field increases and the magnetic pressure can, then, balance the outward force.}
\end{figure*}
%%%%%%%%%%%%%%%%%%%%%%%%%%%%%%%%%%%%%%%%%%%%%%%%%%%%%
\noindent Eq.~(\ref{pa9}) has been solved with the boundary conditions obtained by imposing that, for very large values of the thermodynamic forces there are no privileged directions for $h_{\mu\nu}$  in the thermodynamic space \cite{sonnino3}. In order to take into account the Shafranov displacement, in the limit of high aspect ratio, the {\it magnetic configuration} ${\bf B}$ is written in the form (normalized to $2\pi$)
\begin{equation}\label{sha1}
{\bf B}=F\nabla\phi+\nabla\phi\times\nabla\Psi\qquad {\rm with}\quad F\simeq B_0R_0\quad{\rm and}\quad \Psi(r)\simeq B_0\int_0^rr'/q(r')\ dr'
\end{equation}
\noindent Here $\phi$, $r$, and $\Psi$ are the toroidal angle, the minor radius coordinate and the poloidal magnetic flux, respectively. $B_0$ and $R_0$ denote the intensity of the magnetic field at the magnetic axis of the Tokamak and the major radius of the Tokamak, respectively. In coordinates $(R,\phi ,Z)$, the Shafranov shift $\Delta(r)$ is estimated to be roughly equal to 
\begin{equation}\label{shafranov1}
\left\{ \begin{array}{ll}
R=R_0+\Delta (r)+r\cos\theta& \ \ \mbox{}\\
Z =r\sin\theta& \ \ \mbox{}
\end{array}
\right.
\quad {\rm with}\qquad
\left\{ \begin{array}{ll}
\Delta(r)\simeq\beta_{mag}r^2/R_0\\
\beta_{mag}=2\mu_0 nT/B_v^2
\end{array}
\right.
\end{equation}
\noindent with $\theta$, $\mu_0$ and $B_v$ denoting the poloidal angle, the magnetic permeability constant and the poloidal magnetic field, respectively. 

\noindent Fig.~(\ref{FTU}) shows a comparison between experimental data for fully collisional FTU (Frascati Tokamak Upgrade)-plasmas and the theoretical predictions. In the vertical axis we have the (surface magnetic-averaged) radial electron heat flux, and in the horizontal axis the minor radius of the Tokamak. The lowest dashed profile corresponds to the Onsager (Neoclassical) theory and the bold line to the nonlinear theory (Thermodynamical Field Theory (TFT)) satisfying the TCP, respectively. The highest profile are the experimental data provided by the ENEA C.R. - EUROfusion - in Frascati (Massimo Marinucci). As we can see, the TCP principle is well satisfied in the core of the plasma where plasma is in the collisional transport regime. Towards the edge of the Tokamak, transport is dominated by turbulence.
%%%%%%%%%%%%%%%%%%%%%%%%%%%%%%%%%%%%%%%%%%%%%%%%%%%%%
\begin{figure*}[htb] 
\hspace{0.0cm}\includegraphics[width=10.5cm,height=8cm]{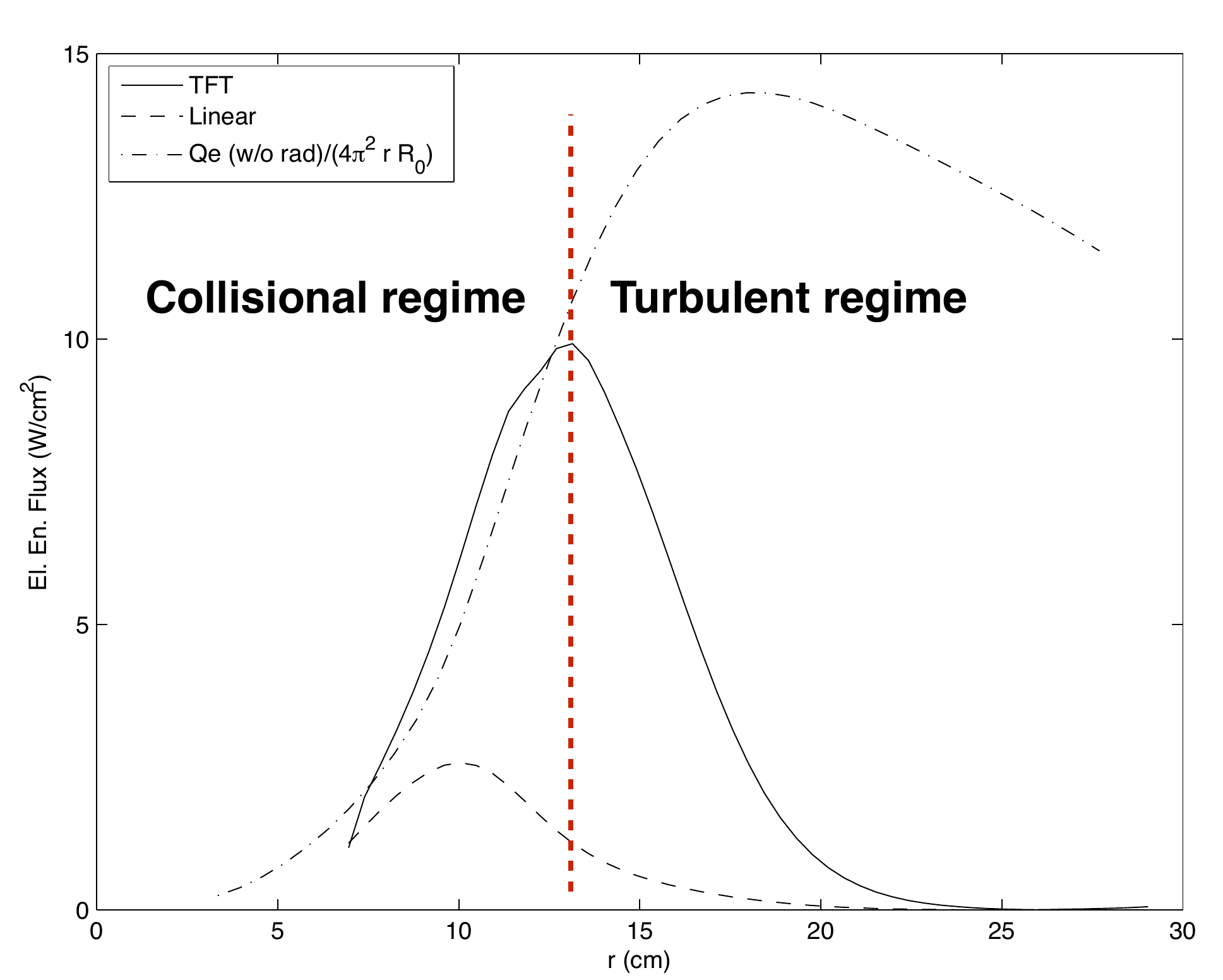}
\caption{ \label{FTU} Electron heat loss in fully collisional FTU-plasmas vs the minor radius of the Tokamak. The highest dashed line is the experimental profile. These data have been provided by Massimo Marinucci from the ENEA C.R. - EUROfusion - in Frascati (Rome - Italy). The bold line is the theoretical profile obtained by the nonlinear theory satisfying the TCP (TFT) and the lowest dashed profile corresponds to the theoretical prediction obtained by Onsager's theory (i.e., by the neoclassical theory).}
\end{figure*}
%%%%%%%%%%%%%%%%%%%%%%%%%%%%%%%%%%%%%%%%%%%%%%%%%%%%%
\noindent We conclude this section by mentioning that close to Onsager's region i.e., $g_{\mu\nu}\simeq L_{\mu\nu}+h_{\mu\nu}$ (with $L_{\mu\nu}$ and $h_{\mu\nu}$ denoting the Onsager matrix and its perturbation, respectively), {\it at the leading order in} $h_{\mu\nu}$ we may transform the closure equation for magnetically confined plasmas (i.e. the last equation in Eqs~(\ref{c8})) in a differential equation, which is covariant under TCT. It is possible to show that the equation to be satisfied by perturbation $h_{\mu\nu}$ is the covariant (under TCT) Laplacian operator constructed with the metric $L_{\mu\nu}$.

\section{Conclusions}\label{conclusions}

We have studied the Lie group associated to the Themodynamic Covariance Principle (TCP). This principle affirms that the nonlinear closure equations must be covariant under the transformations of thermodynamic forces leaving invariant the entropy production and the Glansdorff-Prigogine dissipative quantity. This class of transformations, which form a group, is referred to as the {\it Thermodynamic Coordinate Transformations} (TCT). We have shown that this group, the TCT-group, is a bundle whose base is $diff(\mathbb{RP}^{n-1})$ and the fiber is the space of maps $\mathbb{RP}^{n-1}\rightarrow\mathbb{R}^\times$. The TCT-group may also be split as semidirect product of an abelian normal subgroup and another subgroup of TCT. The irreducible representations of the TCT-group are therefore related to the irreducible representations of these two subgroups.

\noindent According to the TCP, the equivalent character of two representations is warranted iff the fundamental thermodynamic equations are covariant under TCT. In particular, the Lagrangian of a thermodynamic system should be invariant under TCT. As an example of calculation, we have derived the Noether current associated to this TCT-invariance for magnetically confined plasmas in fully collisional transport regime. 

\noindent For Tokamak-plasmas, we have also derived the collisional operator, which guarantees that the TCP is satisfied by the closure transport relations in collisional Tokamak-plasmas. We have shown that, in contrast to the Onsager theory, the theoretical predictions based on the validity of TCP are in fairly good agreement with experiments in the expected region of validity.

\noindent The mathematical study corresponding to the Lie symmetry group associated to this symmetry is under progress. Currently, we are also studying the symmetry-breaking mechanism and the Hamiltonian formulation of problems related to thermodynamic systems out of equilibrium.

\noindent It is worth mentioning that the TCP is actually largely used in a wide variety of thermodynamic processes ranging from non equilibrium chemical reactions to transport processes in tokamak plasmas. As far as we know, the validity of the thermodynamic covariance principle has been verified empirically without exception in physics until now. 

\noindent The influence of nonlinear transformations on fluctuations and a properly modified entropy in
our approach will be subject of future works.

\section{Acknolwedgments}\label{Ack}
G. Sonnino is grateful to Prof. Pasquale Nardone, Prof. Glenn Barnich, and Dr Philippe Peeters from the Universit{\' e} Libre de Bruxelles (ULB) for their suggestions. 
Jarah Evslin is supported by NSFC MianShang Grant 11375201. 
\vskip 0.5truecm

\noindent {\bf Appendix: Splitting of the TCT-group}\label{splitting TCT-group}

\noindent In this Appendix we shall prove the validity of Eq.~(\ref{algebra5}). Denote by $S^{n-1}$ the $n-1$ dimensional unit sphere represented as a $C^{\infty}$ differentiable manifold, which in our case is a submanifold embedded in $\mathbb{R}^{n}$ of the form 
\begin{equation}
\left\Vert \mathbf{x}\right\Vert =1 \label{1}
\end{equation}
\noindent Here the function $\mathbb{R}^{n}\ni\mathbf{x\mapsto}\left\Vert \mathbf{x}
\right\Vert \in\mathbb{R}^{+}$ is some $C^{\infty}(\mathbb{R}^{n})$ function
having also the properties of a norm. For instance
\[
\left\Vert \mathbf{x}\right\Vert =\left[  \sum\limits_{j=1}^{n}(x_{j}
)^{2k}w_{j}\right]  ^{\frac{1}{2k}}~;~\ w_{j}>0,~k=1,2,...\text{ }
\]
Let $\Gamma_{n}^{S}=Diff(S^{n-1})$ be the group of diffeomorphisms of $S^{n-1}$ and let $\Gamma_{n}^{p}\subset Diff(S^{n-1})$ be the subgroup of $\Gamma_{n}^{S}$ that preserves the equivalence relation
$\mathcal{R}$ induced on $S^{n-1}$ by $\mathbf{x,y\in}S^{n-1}$ are equivalent iff $\mathbf{y=\pm x}$.

\begin{remark}
\noindent The quotient space $S^{n-1}/\mathcal{R}$ is a diffeomorphism with the
$n-1$ dimensional projective space ${\mathbb P}^{n-1}$, so $\Gamma_{n}^{p}$ is isomorphic to $Diff({\mathbb {RP}}^{n-1})$. For all map $\mathbf{x}\rightarrow\mathbf{Y(x)}$ where $\mathbf{Y\in}\Gamma_{n}^{S}$ we have
$\mathbf{Y\in}\Gamma_{n}^{P}$ iff
\begin{equation}
\mathbf{Y(-x)=-Y(x)}\label{1.1}
\end{equation}
\end{remark}
\noindent We denote by ${N}^{n}$ the abelian group generated by all
$C^{\infty}(S^{n-1})$ positive functions where the group operation is
defined by multiplication, with the additional symmetry property
\[
f(\mathbf{x})\in {N}^{n}\text{ }if\ f(-\mathbf{x})=f(\mathbf{x})
\]
We also denote by $G^{n} \subset Diff(\mathbb{R}
^{n}\backslash\{0\})$ the TCT-group: the subgroup of the group of
diffeomorphisms of $\mathbb{R}^{n}\backslash\{0\}$ having the
additional homogeneity property
\begin{align}
\mathbb{R}^{n}\backslash\{0\}  &  \ni\mathbf{x\mapsto Y}_{g}(\mathbf{x}
)\in\mathbb{R}^{n}\backslash\{0\}\label{2}\\
\mathbf{Y}_{g}(\lambda\mathbf{x}) &  =\lambda\mathbf{Y}_{g}(\mathbf{x}
);~\lambda\in\mathbb{R},~g\in G^{n}\, \label{2.1}
\end{align}
\noindent We denote by $N^{n}$ the subset (normal subgroup, see below) of $G^{n}\,$ having the form
\begin{equation}
\mathbf{Y}_{g_{{}}}(\mathbf{x})=\mathbf{x~}r_{g}(\mathbf{x})~;~g\in
N^{n}\subset G^{n} \label{2.2}%
\end{equation}
\noindent where $r_{g}(\mathbf{x})$ is a positive $C^{\infty}(\mathbb{R}^{n}
\backslash\{0\})$ homogeneous function
\begin{equation}
r_{g}(\lambda\mathbf{x})=r_{g}(\mathbf{x})>0;~\lambda\in\mathbb{R} \label{2.3}
\end{equation}
\noindent We have the following {\it proposition}

\begin{proposition}
\label{markerPropHnisgroup} $N^{n}$ is a normal abelian subgroup, and for all
$g,~g_{1},~g_{2}~\in$ $N^{n}$ we have
\begin{align}
r_{g_{1}g_{2}}(\mathbf{x})  &  =r_{g_{1}}(\mathbf{x})r_{g_{2}}(\mathbf{x}
)\label{2.4}\\
r_{g^{-1}}(\mathbf{x})  &  =\frac{1}{r_{g}(\mathbf{x})} \label{2.5}
\end{align}
\end{proposition}

\begin{proof}

\noindent The group properties Eqs.(\ref{2.4}, \ref{2.5}) results immediately, by
direct calculation from the general definition of the group product in
$G^{n}$
\[
\mathbf{Y}_{g_{1}g_{2}}(\mathbf{x}):=\left[  \mathbf{Y}_{g_{1}}\circ
\mathbf{Y}_{g_{2}}\right]  (\mathbf{x});~g_{1},~g_{2}~\in G^{n}
\]
and by from the definition Eq.(\ref{2.2}). The abelian character results from
Eq.(\ref{2.4}). In order to prove that $N^{n}$ is a normal subgroup, let
$h\in N^{n}$ and let be $g\in G^{n}$ an arbitrary element of the TCT-group. We have to prove that
\begin{equation}
u:=ghg^{-1}\in N^{n} \label{2.51}
\end{equation}
\noindent or equivalently, to prove that
\begin{equation}
Y_{u}(\mathbf{x})=Y_{ghg^{-1}}(\mathbf{x})=\left[  Y_{g}\circ Y_{h}\circ
Y_{g^{-1}}\right]  (\mathbf{x})=\mathbf{x~}r(\mathbf{x}) \label{2.52}
\end{equation}
\noindent where
\begin{equation}
Y_{h}(\mathbf{x})=\mathbf{x~}r_{h}(\mathbf{x}) \label{2.53}
\end{equation}
with $r_h(\mathbf{x})$ denoting a positive $C^{\infty}(\mathbb{R}^{n}
\backslash\{0\})$ homogeneous function. Let us also denote
\begin{equation}
Y_{z}(\mathbf{x})=\left[  Y_{g}\circ Y_{h}\right]  (\mathbf{x}) \label{2.531}
\end{equation}
\noindent From Eqs.(\ref{2.1}, \ref{2.53}), we get
\begin{equation}
Y_{z}(\mathbf{x})=r_{h}(\mathbf{x})Y_{g}(\mathbf{x}) \label{2.54}
\end{equation}
\noindent and from Eqs.(\ref{2.52}, \ref{2.54}, \ref{2.1}) we find
\begin{align*}
Y_{u}(\mathbf{x})  &  =\left[  Y_{z}\circ Y_{g^{-1}}\right]  (\mathbf{x}
)=\left[  r_{h}Y_{g}\circ Y_{g^{-1}}\right]  (\mathbf{x})=\\
\left[  r_{h}\circ Y_{g^{-1}}\right]  (\mathbf{x})\left[Y_{g}\circ
Y_{g^{-1}}\right]  (\mathbf{x})  &  =r\left[  Y_{g^{-1}}(\mathbf{x})\right]
~\mathbf{x}
\end{align*}
\noindent Observe that $r\left[  Y_{g^{-1}}(\mathbf{x})\right]  $ possesses all the
properties required by Eq.(\ref{2.3})\ which proves Eq.(\ref{2.52}).
\end{proof}
\noindent Let us now denote by $H^{n}$ the subgroup of $G^{n}$ having the properties
\begin{align}
\left\Vert \mathbf{Y}_{h}(\mathbf{x})\right\Vert  &  =\left\Vert
\mathbf{x}\right\Vert \label{2.6}\\
\mathbf{Y}_{h}(-\mathbf{x})  &  =-\mathbf{Y}_{h}(\mathbf{x})\label{2.7}\\
h  &  \in H^{n}
\end{align}
\begin{remark}
Setting in Eq.~(\ref{2.6}) $\left\Vert \mathbf{x}\right\Vert =1$ and by using
Eq.~(\ref{2.7}), we note that the diffeomorphism group $H{n}$ is isomorphic
to the $Diff({\mathbb {RP}}^{n-1})$\thinspace where ${\mathbb P}^{n-1}$ is the $n-1$ dimensional
projective space, since ${\mathbb P}^{n-1}$ can be represented as $S^{n-1}$ with
identified antipodal points.
\end{remark}
\noindent We have the following {\it proposition}
\begin{proposition}
\label{markerPropGnRepresentation} For all $g\in G^{n}$ we have the unique
representation
\begin{align}
g  &  =hg_{N};~g_{N}\in N^{n}~,h\in H^{n}\label{2.8}\\
\mathbf{Y}_{g}(\mathbf{x})  &  =\left[  \mathbf{Y}_{h}\circ\mathbf{Y}_{g_{N}
}\right]  (\mathbf{x}) \label{2.9}
\end{align}
\end{proposition}
\begin{proof}

\noindent Existence of the representation: note that by setting
\begin{align}
\mathbf{Y}_{h}(\mathbf{x})  &  =\frac{\mathbf{Y}_{g}(\mathbf{x})\left\Vert
\mathbf{x}\right\Vert }{\left\Vert \mathbf{Y}_{g}(\mathbf{x})\right\Vert
}\label{2.10}\\
\mathbf{Y}_{g_{N}}(\mathbf{x})  &  =\mathbf{x~}r_{N}(\mathbf{x})~\label{2.11}
\\
\mathbf{~}r_{N}(\mathbf{x})  &  =\frac{\left\Vert \mathbf{Y}_{g}
(\mathbf{x})\right\Vert }{\left\Vert \mathbf{x}\right\Vert } \label{2.12}
\end{align}
\noindent Eq.(\ref{2.9}) is verified and $r_{N}(\mathbf{x})$ has the property
Eq.(\ref{2.3}). In order to prove uniqueness, we consider that in
Eq.(\ref{2.9}) $\mathbf{Y}_{h}\in N^{n}$, with property Eqs.(\ref{2.6},
\ref{2.7}), but otherwise arbitrary, and $\mathbf{Y}_{g_{N}}(\mathbf{x}
)=\mathbf{x~}r(\mathbf{x})$ with $r(\mathbf{x})~$ an arbitrary smooth, homogenous function of zero degree. We rewrite Eq.~(\ref{2.9}), by using Eq.~(\ref{2.1})
\begin{equation}
\mathbf{Y}_{g}(\mathbf{x})=\mathbf{Y}_{h}\left[  \mathbf{x~}r(\mathbf{x}
)~\right]  =r(\mathbf{x})\mathbf{Y}_{h}(\mathbf{x})~ \label{2.13}
\end{equation}
Since $r(\mathbf{x})>0$ we have
\begin{equation}
\left\Vert \mathbf{Y}_{g}(\mathbf{x})\right\Vert =\left\vert r(\mathbf{x}
)\right\vert \left\Vert \mathbf{Y}_{g_{N}}(\mathbf{x})\right\Vert
=r(\mathbf{x})\left\Vert \mathbf{x}\right\Vert \label{2.14}
\end{equation}
\noindent which leads to
\begin{equation}
r(\mathbf{x})=r_{h}(\mathbf{x})=\frac{\left\Vert \mathbf{Y}_{g}(\mathbf{x}
)\right\Vert }{\left\Vert \mathbf{x}\right\Vert } \label{2.15}
\end{equation}
\noindent From Eqs.(\ref{2.15}, \ref{2.14}) we obtain Eq.(\ref{2.10}) so the proof of
uniqueness of the representation Eq.(\ref{2.8}).
\end{proof}

\noindent Irrespective to the choice of the norm in the definition of the subgroup
$H^{n}$, we may easily convince ourselves that they are all equivalent up to a group isomorphism.

\noindent For easy reference, we recall the semidirect product definition and properties \cite{wiki},
\cite{Robinson}, \cite{wolfram}, \cite{eom}
\begin{theorem}
\label{markerTheoremWiki}Let $N,~H$ subgroups of the group $G$,~where $\ N$ is
a normal subgroup. Then the following statements are equivalent: $a)$ \ $G=NH$
and $N\cap H=\{e\}$ . $b)$~\ For all $g\in G$ there exists an unique
representation $g=nh$ with $n\in N$ and $h\in H$. $c)$~\ For all \ $g\in G$
there exists an unique representation $g=hn$ with $n\in N$ and $h\in H$. d).
The natural embedding $i:H\rightarrow G$, composed with the natural projection
$p:G\rightarrow G/N$, yields an isomorphism $\psi:H\rightarrow G/N$ ,
\ $\psi=p\circ i$ with inverse \ $\widehat{\chi}:G/N\rightarrow H$. $\ e)$
There exists a homomorphism $\chi:G\rightarrow H$ that is the identity on $H$
and whose kernel is $N$.
\end{theorem}
\noindent If one of the above properties are verified, $G$ is said to split in a semidirect product of the subgroups $H$ and normal subgroup $N$. In this case the representations of the group $G$ are related to the representations of the subgroups $H$ and $N$. 

\noindent By using the Theorem \ref{markerTheoremWiki} and Propositions \ref{markerPropGnRepresentation}, \ref{markerPropHnisgroup}, we finally get

\begin{theorem}
\label{markerTheoremTCPsplitting} The TCT-group $G^{n}$ is a semidirect
product of the abelian normal subgroup \ $N^{n}$ and the subgroup $H^{n}$
\begin{equation}
G^{n}=N^{n}~\rtimes H^{n} \label{2.16}
\end{equation}
\end{theorem}

\noindent {\bf Appendix: Calculation of the Noether Current}\label{Noether}

\noindent We sketch in some detail the derivation of Eq.~(\ref{n4}). From Action ({\ref{pa5}), we have
\begin{equation}\label{2A1}
I=\int\Bigl[ R-(\Gamma^\lambda_{\alpha\beta}-
{\tilde\Gamma}^\lambda_{\alpha\beta})S^{\alpha\beta}_{\lambda}
\Bigr]\sqrt{g}\ \! {d^{}}^n\!X=\int L\ \! {d^{}}^n\!X
\end{equation}
\noindent with
\begin{equation}\label{2A2}
L\equiv \Bigl[R-(\Gamma^\lambda_{\alpha\beta}-
{\tilde\Gamma}^\lambda_{\alpha\beta})S^{\alpha\beta}_{\lambda}
\Bigr]\sqrt{g}
\end{equation}
\noindent For easy reference, we report the expression of the Noether current
\begin{equation}\label{2A3}
j^\mu_\alpha=\frac{\partial L}{\partial\Phi^A_{,\mu}}{\mathcal L}_{{\mathbf\xi}_\alpha}\Phi^A-L\xi^\mu_\alpha\quad \quad {\rm with}\quad \Phi^A=(g_{\mu\nu},\Gamma^\kappa_{\mu\nu})\quad ;\quad (\alpha=1,\cdots ,N)
\end{equation}
\noindent The first term appearing in Noether's current (\ref{2A3}), i.e. $\partial L/\partial\Phi_\mu^A$, is computed directly from Eqs~(\ref{pa7}) and Eq.~(\ref{2A2}). We get
\noindent 
\begin{eqnarray}\label{2A4}
\!\!\!\!\!\!\!\!\frac{\partial L}{\partial g_{\kappa\nu,\lambda}}&=&\Bigr[\frac{1}{2}g^{\alpha\kappa}S^{\nu\kappa}_\alpha+\frac{1}{2}g^{\alpha\nu}S^{\kappa\lambda}_\alpha-\frac{1}{2}g^{\alpha\lambda}S^{\kappa\nu}_\alpha+\frac{1}{2\sigma}X^\alpha X^\lambda S_\alpha^{\nu\kappa}-\frac{X^\kappa X^\lambda}{2(n+1)\sigma}S^{\alpha\nu}\nonumber\\
\!\!\!\!\!\!\!\!&-&\frac{X^\nu X^\lambda}{2(n+1)\sigma}S^{\alpha\kappa}+\frac{1}{2\sigma}X^\alpha X^\lambda \Bigl(\Gamma_{\alpha\beta}^\kappa g^{\beta\nu}+\Gamma_{\alpha\beta}^\nu g^{\alpha\kappa}\Bigr)
\nonumber\\
\!\!\!\!\!\!\!\!&-&\frac{X^\lambda (X^\nu+ X^\kappa)}{2(n+1)\sigma}\Bigl(\Gamma^\alpha_{\alpha\beta}g^{\beta\kappa}+\Gamma_{\alpha\beta}^\kappa g^{\alpha\beta}\Bigr)-\frac{1}{2\sigma}\Bigl(X^\beta X^\lambda g^{\kappa\nu}-\frac{X^\nu X^\lambda}{n+1}g^{\beta\kappa}\nonumber\\
\!\!\!\!\!\!\!\!&-&\frac{X^\kappa X^\lambda}{n+1}g^{\beta\nu}\Bigr)\Gamma_{\alpha\beta}^\alpha\Bigr]\sqrt{g}\nonumber\\
\!\!\!\!\!\!\!\!\frac{\partial L}{\partial\Gamma^\eta_{\kappa\eta,\lambda}}& =&\Bigl(\frac{1}{2}g^{\kappa\lambda}\delta^\nu_\eta+\frac{1}{2}g^{\nu\lambda}\delta^\kappa_\eta-g^{\kappa\nu}\delta^\lambda_\eta\Bigr)\sqrt{g}
\end{eqnarray}
\noindent The Lie derivatives of the fields $\Phi^A=(g_{\mu\nu}, \Gamma^\mu_{\nu\kappa})$ read
\begin{eqnarray}\label{2A5}
\!\!\!\!\!\!\!\!{\mathcal L}_{\delta X_{(\alpha)}\epsilon^\alpha}g_{\mu\nu}&=&\bigl[\partial_\mu\bigl(\delta X^\lambda_{(\alpha)}\epsilon^\alpha\bigr)\Bigr]g_{\lambda\nu}+\bigl[\partial_\nu\bigl(\delta X^\lambda_{(\alpha)}\epsilon^\alpha\Bigr]g_{\lambda_\mu}+\delta X^\lambda_{(\alpha)}\epsilon^\alpha g_{\mu\nu,\lambda}
\\
\!\!\!\!\!\!\!\!{\mathcal L}_{\delta X_{(\alpha)}\epsilon^\alpha}\Gamma^\mu_{\nu\kappa}&=&\Bigl[\partial_\kappa (\delta X^\eta_{(\alpha)}\epsilon^\alpha)\Bigr]\Gamma^\mu_{\nu\eta}+\Bigl[\partial_\nu (\delta X^\eta_{(\alpha)}\epsilon^\alpha)\Bigr]\Gamma^\mu_{\eta\kappa}-\Bigl[\partial_\beta (\delta X^\mu_{(\alpha)}\epsilon^\alpha)\Bigr]\Gamma^\beta_{\nu\kappa}+\delta X^\lambda_{(\alpha)}\epsilon^\alpha\Gamma^\mu_{\nu\kappa,\lambda}\nonumber\\
\!\!\!\!\!\!\!\! &+&\partial^2_{\nu\kappa}\bigl(\delta X^\mu_{(\alpha)}\epsilon^\alpha\bigr)\nonumber
\end{eqnarray}
\noindent with displacement $\delta{\mathbf X}_{(\alpha)}$ coinciding with $\xi_\alpha$ (i.e., $\delta{\mathbf X}_{(\alpha)}\equiv{\mathbf\xi}_\alpha$). Note that, in literature, the Lie derivative of the affine connection, is referred to as the {\it pseudo-Lie derivative} due to the presence of the last term in the second equation of Eqs~(\ref{2A5}) (i.e., the second derivative of the infinitesimal vector $\delta X^\mu_{(\alpha)}\epsilon^\alpha$). Now, taking into account Eqs~(\ref{algebra8})-(\ref{algebra11}), in the limit of $\sigma\gg 1$, we get Eq.~(\ref{n4}).

\bigskip
\end{document}